\documentclass[aps,pra,showpacs,twoside,twocolumn,10pt,nofootinbib]{revtex4-1}
\usepackage[colorlinks=true, citecolor=red, urlcolor=blue ]{hyperref}
\usepackage{epsfig,amssymb,amsfonts,amsmath,bm,subfigure,mathtools,amsthm,braket,soul,enumitem,xcolor,physics,graphics,graphicx}
\UseRawInputEncoding
\usepackage[normalem]{ulem}
\usepackage{comment}
\usepackage{epstopdf}
\usepackage{multirow}
\usepackage{cancel}
\usepackage{soul}

\newtheorem{theorem}{Theorem}

\newtheorem{definition}{Definition}
\newtheorem{proposition}{Proposition}

\newcommand{\sm}[1]{{\color{blue}#1}}

\newcommand{\ph}[1]{{\color{red}#1}}

\begin{document}

\title{More nonlocality with less incompatibility in higher dimensions: \\
Bell vs prepare-measure scenarios}
\author{Sudipta Mondal$^{1}$, Pritam Halder$^{1}$, Saptarshi Roy$^{2}$, Aditi Sen (De)$^{1}$}
\affiliation{$^1$Harish-Chandra Research Institute, A CI of Homi Bhabha National Institute, Chhatnag Road, Jhunsi, Prayagraj 211019, India\\
\(^2\)  QICI Quantum Information and Computation Initiative, Department of Computer Science,\\
The University of Hong Kong, Pokfulam Road, Hong Kong }

\begin{abstract}

Connecting incompatibility in measurements with the violation of local realism is one of the fundamental avenues of research. For two qubits, any incompatible pair of projective measurements can violate Clauser-Horne-Shimony-Holt (CHSH) inequality for some states, and there is a monotonic relationship between the level of measurement incompatibility (projective) and the violation. However, in the case of two qutrits, we exhibit that the violation of the Collins-Gisin-Linden-Massar-Popescu (CGLMP) inequality responds non-monotonically with the amount of incompatibility; we term this {\it more nonlocality with less incompatibility}. Furthermore, unlike in the CHSH case, the maximally violating state in higher dimensions depends on the amount of measurement incompatibility. We illustrate that similar patterns can also be observed in an experimentally viable interferometric measuring technique. In such a measurement scenario, we provide an explicit example of incompatible (not jointly measurable) measurements that do not violate the CGLMP inequality for any shared quantum state. We extend our study of incompatibility in the prepare and measure scenario, focusing on quantum random access codes (QRACs). Surprisingly, we show that the monotonicity of average success probability with measurement incompatibility does not hold for higher dimensions, as opposed to two dimensions, even though the maximum probability of QRAC behaves monotonically with incompatibility.




\end{abstract}

\maketitle

\section{Introduction}
\label{sec:intro}

One of the striking non-classical characteristics of quantum physics is the existence of incompatible measurements that cannot be measured jointly and simultaneously \cite{gunhe2023incompatibility}. Similar to other facets of nonclassicality like entanglement \cite{hhhh} and nonlocality \cite{Bell1964, Bruner2014bell}, measurement incompatibility has been shown to be advantageous in various information processing tasks such as state discrimination \cite{Barnett2009, Bae2015,sd1,sd2,sd3}, communication games including random access code (RAC) \cite{Ambainis2002,rac2,rac3} as well as demonstrating violations of Bell inequalities \cite{incombell, incombell1,Bene2018mi}  and quantum steering \cite{steering2, steering3, steering1}. As a result, it can be identified as a resource for which a resource theoretic framework has been developed \cite{buscemirt1}.


Exploring the connection between measurement incompatibility and other non-classical behaviors, like the violation of local realism, has become a pertinent avenue of research in recent years.  For the Clauser-Horne-Shimony-Holt (CHSH) inequality \cite{chsh}, it has been demonstrated that incompatible projective measurements necessarily lead to the violation of local realism \cite{Cirelson1980}. In a stronger sense, the obtained violation is monotonic with the amount of incompatibility in the measurements used to extract the statistics (at least in the projective case) \cite{Cirelson1980, wolf2009}. However, for generalized positive-operator-valued measurements (POVMs), such one-to-one correspondence fails to hold since examples of incompatible POVMs are found that are incapable of providing statistics for any shared state, leading to the violation of the CHSH inequality \cite{priya_ghosh2023, incombell, incombell1}. These findings, primarily focused on two-dimensional systems, highlight the need to explore similar connections in higher-dimensional scenarios, which remain underrepresented till date.

To investigate these subtleties more deeply, we focus on determining their relationship for higher dimensional Bell inequalities, particularly the Collins-Gisin-Linden-Massar-Popescu (CGLMP) inequalities  \cite{Collins2002bell}, which have already been demonstrated experimentally \cite{Lo2016experimental, Zhang2024experiment}. These inequalities were formulated to capture quantum correlations present in two qudits with greater efficiency and to provide enhanced robustness against noise \cite{Collins2002bell}. Even while the CGLMP inequality has been well examined in the literature in several scenarios \cite{cglmp2, cglmp3,cglmp4,cglmp5,roy2024robustness},  
there are still a number of uncharted areas that need to be addressed. Furthermore, qudit systems have been certified to offer more advantages compared to their qubit counterparts across a wide range of quantum information tasks including the robust test of quantum steering \cite{marciniak2015,Srivastav2022}, quantum simulation and computation \cite{Neeley2009,Kaltenbaek2010}, quantum secure key rates \cite{bechmann2000,cerf_qkd_prl_2002,nikolopoulos2006,sheridan2010,Sasaki2014,Bouchard2018,araujo2023},  quantum telecloning \cite{nagali2010,Bouchard2017}, and quantum thermal machines \cite{ correra2014,santos2019,Dou2020,ghosh2022, konar2023}. This growing number of works showing the necessity of increasing dimensions also makes it imperative to establish links between measurement incompatibility and violation of the Bell inequalities for higher dimensions \cite{Collins2002bell}.

 In this work, we provide a conclusive proof that incompatible projective measurements are necessary and sufficient to violate the CGLMP inequality with two qutrits. In contrast to the CHSH situation \cite{chsh}, we find that the violation of the CGLMP inequality is \textit{nonmonotonic} with the increase of incompatibility in projective measurements quantified via commutation based incompatibility measure \cite{Mordasewicz2022}. This leads to one of the central observations of our work, where for the CGLMP inequality, we report: \textit{more nonlocality with less incompatibility.} Interestingly, we also observe such nonmonotonic behavior when one also considers random projective measurements.
Another curious feature that stems from our analysis is the non-uniqueness of the maximally violating state. Unlike the CHSH inequalities, for which the maximally entangled state is the maximally violating state for any incompatible measurements, this is not true for the CGLMP inequality. Specifically, in this case, the maximally violating state depends on the incompatibility of the measurements employed to obtain the Bell statistics. We analyze non-monotonicity and nonuniqueness properties when the measurement choices are restricted to those that can be implemented in a photonic experimental setup \cite{Collins2002bell,Lo2016experimental, Zhang2024experiment}. {\color{black} 
We also provide an explicit example of POVMs that cannot be measured together yet do not violate the CGLMP inequality, emphasizing the inadequacy of measurement incompatibility to demonstrate a violation for higher dimensional Bell inequalities.}


On the application front, we investigate the role of measurement incompatibility in the prepare and measure (PM) scenarios \cite{Wiesner1983,Poderini2020,Sophie2024} in higher dimensions $(d > 2)$. This work focuses on a particular case of PM scenario, namely quantum random access codes (QRACs) \cite{Ambainis2002,rac2,rac3} -- the sender is restricted in sending quantum systems that do not share any prior correlation with the receiver while only classical correlations can be shared beforehand. 
{\color{black}
The examination of RACs in higher dimensions is particularly important because of the pre-existing connections between the Bell-CHSH game and RAC in terms of XOR games \cite{xor1,xor2} in $d=2$. Furthermore, other works \cite{archan, Sophie2024} also highlight an isomorphism between the Bell-CHSH game and RAC tasks.  The above results justify the significance of our investigation exploring the connection between RACs and nonlocality in higher dimensional settings with respect to measurement incompatibility.
}

In contrast to the Bell setting, the quantity of interest is the success probability that the receiver will obtain an outcome given that the inputs are received by the sender and the receiver as classical messages. 
Our investigation conclusively demonstrates that in dimensions greater than two, nonlocality and the success probability in the QRAC task display \textit{inequivalent} responses to measurement incompatibility. In particular, the upper bound of maximal average success probability in the QRAC task is monotonic with the amount of measurement incompatibility for projective measurements. However, our analysis also reveals the existence of projective measurements for which we get greater average success probability in RAC game with less incompatible measurements. Furthermore, we determine the maximum noise strength that can maintain the quantum advantage in the QRAC game and violate the CGLMP inequality, which turns out to be inequivalent in higher dimensions, in contrast to the two-dimensional case. Nonetheless, we construct measurement strategies with identical noise robustness in CGLMP and QRAC scenarios.

The contents of the paper are organized as follows. We investigate the connections between violations of the CGLMP inequality and measurement incompatibility in Sec. \ref{sec:incomvscglmp} when the measurement is restricted to projective measurements. Sec. \ref{sec:interferometric} discusses whether the findings for random projective measurements hold for measurements implementable in a photonic set-up. Before concluding \ref{sec:conclusion}, we establish a connection between measurement incompatibility and prepare-measure scenario, especially QRACs in Sec. \ref{sec:rac_incom}.

\section{Incompatibility vs. violations of CGLMP inequality: projective measurements}
\label{sec:incomvscglmp}

The central objective of this section is to investigate the connection between the violation of Bell inequality and the measurement incompatibility for two arbitrary \(d\)-dimensional systems \((d>2)\).  In particular, we consider the CGLMP inequality \cite{Collins2002bell} as the quantifier for exhibiting higher dimensional nonlocality, while the incompatibility of two measurements $M^1\equiv\{M^1_i\}$ and $M^2\equiv\{M^2_j\}$, acting on $d$-dimensional complex Hilbert space, $\mathbb{C}^d$, is defined as the Schatten $p$-norms \cite{Mordasewicz2022} of commutator between the elements of the above set of measurements. It  can mathematically be written  as
\begin{eqnarray}
    \mathcal{I}_p(M^1,M^2)= \sum_{j=1}^{d_1} \sum_{k=1}^{d_2} \|[M_j^1,M_k^2]\|_p,
\label{eq:incom_p_norm_1}
\end{eqnarray}
where for an arbitrary projector $A$, $||A||_p=\left[\tr[\sqrt{A^\dagger A}]^p\right]^{1/p}$ and $p\in[1,\infty]$ (see Appendix \ref{sec:incom} for its properties in detail). Here, $i,j$ denotes the outcomes of measurement $M^1$ and $M^2$ respectively. {\color{black}
The above notion of incompatibility generalizes and formalizes the idea of incompatibility of two observables $O^1$ and $O^2$ expressed in terms of their commutator $[O^1,O^2]$. In particular, if $O^1$ and $O^2$ are qubit observables with eigenvalues $\pm 1$, then we have
\begin{eqnarray}
    ||[O^1,O^2]||_p = \mathcal{I}_p(\mathcal{S}_{O^1},\mathcal{S}_{O^2}),
\end{eqnarray}
where $\mathcal{S}_{O^{1(2)}}$ is the set of all the projectors in the spectral decomposition of $O^{1(2)}$. 
Formally, when $O^1$ and $O^2$ are two $d$-dimensional observables with $1 \leq ||O^{1(2)}||_p \leq \lambda_{\max}^{1(2)}$, we have
\begin{eqnarray}
    ||[O^1,O^2]||_p \leq \lambda_{\max}^1 \lambda^2_{\max} \mathcal{I}_p(\mathcal{S}_{O^1},\mathcal{S}_{O^2}),
\end{eqnarray}
The above results can be obtained by using the spectral decomposition $O^{1(2)}= \sum_{i = 1}^d \lambda_i^{1(2)} \Pi^{1(2)}_i$ and the convexity of the Schatten $p$-norm.
From an operational perspective, $\mathcal{I}_p$ provides a necessary and sufficient criterion of joint measurability for two sets of projective measurements. 
In particular, we prove it in the Appendix. \ref{sec:incom} that any two projective measurements, $\Pi^1 = \{\Pi^1_i\}$ and $\Pi^2 = \{\Pi^2_j\}$ are jointly measurable if and only if $\mathcal{I}_p(\Pi^1,\Pi^2)$ vanishes identically.
}

Let us now briefly describe the CGLMP inequality and the corresponding form of the CGLMP operator required for our study.

\subsubsection*{CGLMP inequality}
We investigate the violations of local realism in higher dimensions $(d > 2)$ using the CGLMP inequality \cite{Collins2002bell}. It is a generalization of the CHSH inequality \cite{chsh} where measurements at each party can support an arbitrarily high number of outcomes.
Technically, it generalizes the standard 2 (party)-2 (settings)-2 (outcomes) configuration to a $2-2-d$ scenario. For a theory that satisfies local realism, the CGLMP expression \cite{Collins2002bell} $\chi_d$ satisfies the bound,
\begin{widetext}
\begin{eqnarray}
    \chi_{d}&=& \sum_{k=0}^{[\frac{d}{2}]-1}\bigg(1-\frac{2k}{d-1}\bigg)\{P(A^{1}=B^{1}+k) + P(B^{1}=A^{2}+k+1) + P(A^{2}=B^{2}+k) + P(B^{2}=A^{1}+k)\nonumber\\&& - [P(A^{1}=B^{1}-k-1) + P(B^{1}=A^{2}-k) + P(A^{2}=B^{2}-k-1) + P(B^{2}=A^{1}-k-1)]\}\leq2,
    \label{eq:cglmp}
\end{eqnarray}
\end{widetext}
where \(P(A^{a}=B^{b}+k)\) is the probability that the measurement outcomes of two distant parties, Alice and Bob differ by \(k\) mod \(d\), and is given by
\begin{eqnarray}
    P(A^{a}=B^{b}+k)&=&\sum_{j=0}^{d-1}P(A^{a}_{j+k\mod d},B^{b}_j).\nonumber\\
    \label{eq:single_probability}
\end{eqnarray}
Like in the CHSH inequality, in this case too, the maximum value of \(\chi_d\) is \(4\) while for local hidden variable theories (LHVT) the maximum that can be achieved is \(2\) \cite{Collins2002bell}.

We now intend to cast the CGLMP expression in an operator form to facilitate the examination of its connection with measurement incompatibility. Given a bipartite state, $\rho_{AB}\in \mathcal{C}^d\otimes\mathcal{C}^d$, let Alice $(A)$ and Bob $(B)$ be two spatially separated observers, allowed to perform two local measurements $A^a$ $(a=1,2)$ and $B^b$ $(b=1,2)$ respectively. Each of these measurements can have \(d\) possible outcomes. Therefore, the joint probability, $P(j,k|a,b)$ of detecting the output $j$ in Alice's side and $k$ in Bob's side can be expressed as
\begin{eqnarray}
    P(j,k|a,b) = P(A_j^a,B_k^b) = \Tr[A_j^a\otimes B_k^b\rho_{AB}],
\end{eqnarray}
where the subscripts, $j=0,1,\ldots,d-1$ and $k=0,1,\ldots,d-1$ denote the outcomes of $A^a$ and $B^b$ respectively. 

 In quantum theory, from Eq.~\eqref{eq:cglmp}, we can obtain the maximal value of the CGLMP expression by maximizing over all possible shared quantum states between Alice and Bob, i.e, 
\begin{eqnarray}
     \chi_d = \max_{\rho_{AB}}\Tr[\mathbf{C}_d \rho_{AB}],
     \label{eq:op_cglmp}
\end{eqnarray}
where the CGLMP operator $\mathbf{C}_d$ is given by
\begin{widetext}
\begin{eqnarray}
   \nonumber \mathbf{C}_d &=& \sum_{k=0}^{[\frac{d}{2}]-1}\bigg(1-\frac{2k}{d-1}\bigg)\sum_{j=0}^{d-1}\big(A^1_{j+k \mod d}\otimes B_j^1 + A^2_{j}\otimes B_{j+k+1 \mod d}^1 + A^2_{j+k\mod d}\otimes B_{j}^2 + A^1_{j}\otimes B_{j+k\mod d}^2\\ &-& A^1_{j-k-1 \mod d}\otimes B_j^1 - A^2_{j}\otimes B_{j-k \mod d}^1 - A^2_{j-k-1\mod d}\otimes B_{j}^2 - A^1_{j}\otimes B_{j-k-1\mod d}^2\big).
   \label{eq:cglmp_operator}
\end{eqnarray}
\end{widetext}
Below, we demonstrate the connection between the incompatibility of projective measurements and violation of local realism in the bipartite scenario.

\subsection{Is measurement incompatibility necessary and sufficient for violations of the CGLMP inequality?}
Let us now establish whether incompatible projective measurements are sufficient for violating the qudit CGLMP inequality. For this purpose, we consider arbitrary pairs of projective measurements on Alice and Bob's side, denoted by $(A^1, A^2)$ and $(B^1, B^2)$ respectively such that they are strictly incompatible, i.e,
\begin{eqnarray}
    \mathcal{I}_p(A^1, A^2) > 0, ~~\mathcal{I}_p(B^1, B^2) > 0 ~\forall p, 
\end{eqnarray}
where $\mathcal{I}_p$ represents the Schatten p-norm of incompatibility defined in Eq. \eqref{eq:incom_p_norm_1}, see Appendix for details. Note that from henceforth in this manuscript, we would denote $\mathcal{I}:=\mathcal{I}_\infty$, and it would be the measure of choice throughout the manuscript unless mentioned otherwise.
We parametrize the local measurements on Alice and Bob's side as
\begin{eqnarray}
A^a_j=U_{A^a}\ketbra{j}{j}U_{A^a}^\dagger, ~~   B^b_i=U_{B^b}\ketbra{k}{k}U_{B^b}^\dagger,
\label{eq:unitary_rand}
\end{eqnarray}
where $a,b = 1,2$ and $j,k = 0,1,2, \ldots  d-1$.
Here, each $U_{A^a(B^b)} \in U(d)$ 
can be parameterized by $d^2$ real parameters. 
We group these $ 4d^2$ measurement parameters in a vector $\bm{\mu}$.  
Now, let us define a set $\Theta$ containing all the possible measurement parametrization, $\Theta = \{ \bm{\mu}\}$. 
So, each instance of $\Theta$ defines one distinct measurement strategy employed by Alice and Bob. 
For a given strategy $\bm{\mu}' \in \Theta$, one can construct the CGLMP operator, $\mathbf{C}_d(\bm{\mu}')$ following Eq.~\eqref{eq:cglmp_operator}. According to Eq.~\eqref{eq:op_cglmp}, the maximal violation of the CGLMP inequality can be achieved when we take $\rho_{AB}$ as the eigenstate of $\mathbf{C}_d$ corresponding to the maximum eigenvalue. We now concern ourselves with a subset $\Theta_{\epsilon_A,\epsilon_B}$ of $\Theta$ $(\Theta_{\epsilon_A,\epsilon_B} \subset \Theta)$ that corresponds to only those local measurements of Alice and Bob that are incompatible by an amount $\epsilon_A > 0$ and $\epsilon_B>0$ respectively, i.e., $\mathcal{I}(A^1,A^2)= \epsilon_A$ and $\mathcal{I}(B^1,B^2) = \epsilon_B$. With this, we are set to formally declare what we mean by the sufficiency of measurement incompatibility in violating the CGLMP inequality.
\begin{definition}
    Measurement incompatibility is sufficient for the violation of the CGLMP inequality if $\chi_d>2$ for all $\bm{\eta} \in \Theta_{\epsilon_A,\epsilon_B}$ for any $\epsilon_A, \epsilon_B >0$.
\end{definition}
\noindent This is equivalent to checking whether $\chi_d^{\min}>2$, where 
\begin{eqnarray}
    \chi_d^{\min} := \min_{\bm{\eta} \in \Theta_{\epsilon_A,\epsilon_B}} \max_{\rho_{AB}} \text{Tr}[\mathbf{C}_d(\bm{\eta}) \rho_{AB}]. 
    \label{eq:idmin}
\end{eqnarray}
Since $\max_{\rho_{AB}} \text{Tr}[\mathbf{C}_3(\bm{\eta}) \rho_{AB}]$ is equivalent to finding out the maximal eigenvalue of $\mathbf{C}_3(\bm{\eta})$, the optimization dimension, i.e., the number of free parameters over which the minimization takes place is $4d^2$. This, in general, is a hard problem, but for $d = 3$, one is left with a minimization of $36$ parameters that can be tackled using standard optimization techniques. Here, the parametrization of $U(3)$ unitaries is given in Appendix ~\ref{sec:d3_local_unitary}. In particular, in our case, we used an algorithm that relies on an improved stochastic ranking evolution strategy ($\texttt{ISRES}$) \cite{Runarsson2005} for global optimization and yields $\chi_3^{\min}>2$ for all $\epsilon_A,\epsilon_B>0$. Specifically, we find $\chi_3^{\min}>2+ m_p$, where $m_p \sim 10^{-5}$ is the numerical precision used in our analysis. At any rate, the upshot of our analysis can be encapsulated in the following:
\begin{proposition}
    Incompatible projective measurements are sufficient to violate the CGLMP inequality for two qutrits.
    \label{prop:sufficient}
\end{proposition}

The natural follow-up question is whether incompatibility in measurement is necessary.
Although the necessity of incompatibility has been highlighted at the level of hidden variable model \cite{steering2}, in Appendix \ref{app:necessary} we provide a direct proof of the same by looking at the form of the CGLMP expression itself. This leads to the following Proposition.
\begin{proposition}
    Incompatible projective measurements are necessary to violate the  CGLMP inequality for two qutrits.
    \label{prop:necessary}
\end{proposition}

Finally, combining Propositions \ref{prop:sufficient} and  \ref{prop:necessary}, we arrive at our first result.
\begin{theorem}
    Incompatible projective measurements are necessary and sufficient to obtain violation in the CGLMP inequality for two qutrits. 
\end{theorem}

\subsection{More nonlocality with less incompatibility}

We begin by defining $\chi_d^{\max}$ which denotes the maximal violation of the CGLMP inequality when the measurements performed by any one of the parties are incompatible by a fixed amount, \({\tt I}\), mathematically,
\begin{eqnarray}
    \chi_d^{\max}({\tt I} ) := \max_{\bm{\eta} \in \Theta_{\mathcal{I}(B^0,B^1)={\tt I}}} \max_{\rho_{AB}} \text{Tr}[\mathbf{C}_d(\bm{\eta}) \rho_{AB}], 
    \label{eq:idmax}
\end{eqnarray}
where $\Theta_{\mathcal{I}(B^1,B^2)={\tt I}} \subset \Theta$ such that all the measurement strategies $\eta \in \Theta_{\mathcal{I}(B^1,B^2)={\tt I}}$ has a fixed incompatibility value for measurements at Bob's side, $\mathcal{I}(B^1,B^2) = {\tt I}$. Note that one could have equivalently defined $\chi_d^{\max}({\tt I})$ by fixing the incompatibility of measurements performed by Alice. This equivalence is guaranteed by the symmetry of exchanging the parties $A$ and $B$ without changing the physical setting.  {\color{black} Before discussing the $d=3$ case, let us briefly describe the qubit scenario. For the CHSH inequality, the maximal violation with two outcome projective measurements on both Alice and Bob's side can be expressed as
\begin{eqnarray}
    \chi_2 = 2\sqrt{1 + \frac{\mathcal{I}_p(A^1,A^2)\mathcal{I}_p(B^1,B^2)}{4.2^\frac2p}}.
    \label{eq:chigenI}
\end{eqnarray}
See Appendix \ref{app:chi2(I)} for details of the calculation. 
Note that, like in the rest of the manuscript, we work with $p = \infty$.
When the incompatibility of Bob's measurements is fixed, $\mathcal{I}(B^1,B^2) = {\tt I}$ then from Eq. \eqref{eq:chigenI} it follows that the maximal violation is obtained for maximally incompatible measurements by Alice, $\mathcal{I}(A^1,A^2) = 2$. The maximal CHSH violation with $\mathcal{I}(B^1,B^2) = {\tt I}$ is then given by
\begin{eqnarray}
    \chi_2 = 2\sqrt{1 + \frac{1}{2}{\tt I}}.
    \label{eq:chi2}
\end{eqnarray}
By varying the incompatibility $\mathcal{I}(B^1,B^2) = {\tt I}$, $\chi_2$ increases monotonically, as shown in Fig.~\ref{fig:cglmp_incom_one_para}  and the maximal violation is achieved when $B^1$ and $B^2$ are the set of eigenvectors of MUBs. 
}

For analyzing, $\chi_3^{\max}({\tt I})$, we again employ the same numerical procedure used for computing $\chi_3^{\min}$. Specifically, we compute $\chi_3^{\max}$ by fixing $\mathcal{I}(B^1,B^2) = {\tt I} \in (0,3\sqrt{2}\approx4.24)$. Note that \(3\sqrt{2}\) is the algebraic maximum for \({\tt I}\) in case of three dimension. The behavior of $\chi_3^{\max}$ reveals some contrasting behaviors that cannot be found for two qubits as will be highlighted more in the following. 
\begin{enumerate}
    \item The upper bound of CGLMP violation for a given ${\tt I}$, $\chi_3^{\max}({\tt I})$, is non-monotonic with respect to the amount of incompatibility ${\tt I}$, as shown in Fig.~\ref{fig:incom_viol}.

    \item The most incompatible measurement pair does not violate CGLMP inequality maximally. In particular, the maximum value of incompatibility in $d=3$ is attained for MUBs \cite{Mordasewicz2022} and is given by $\mathcal{I}(A^1,A^2)=\mathcal{I}(B^1,B^2)=3\sqrt{2}$ while the corresponding violation of the CGLMP inequality is given by 
    \begin{eqnarray}
        \chi_d^{\max}({\tt I} = 3\sqrt{2})\approx 2.8060.
    \end{eqnarray}
    It is strictly less than the maximal CGLMP violation of 
    \begin{eqnarray}
       ~~~~~ \max_{\tt I} \chi_3^{\max}({\tt I}) = \chi_3^{\max}({\tt I}^*) = \sqrt{11/3}\approx 2.9149,
    \end{eqnarray}
as obtained by maximally violating state, different from maximally entangled state \cite{acin2002}.
    \item The CGLMP inequality is maximally violated for a configuration of local measurements satisfying $\mathcal{I}(A^1,A^2)=\mathcal{I}(B^1,B^2)= {\tt I}^*= 3.924$. This illustrates an interesting feature where we get \textit{more violation of the CGLMP inequality with non-maximally incompatible projective measurements.}
\end{enumerate}

\begin{figure}
\includegraphics [width=\linewidth]{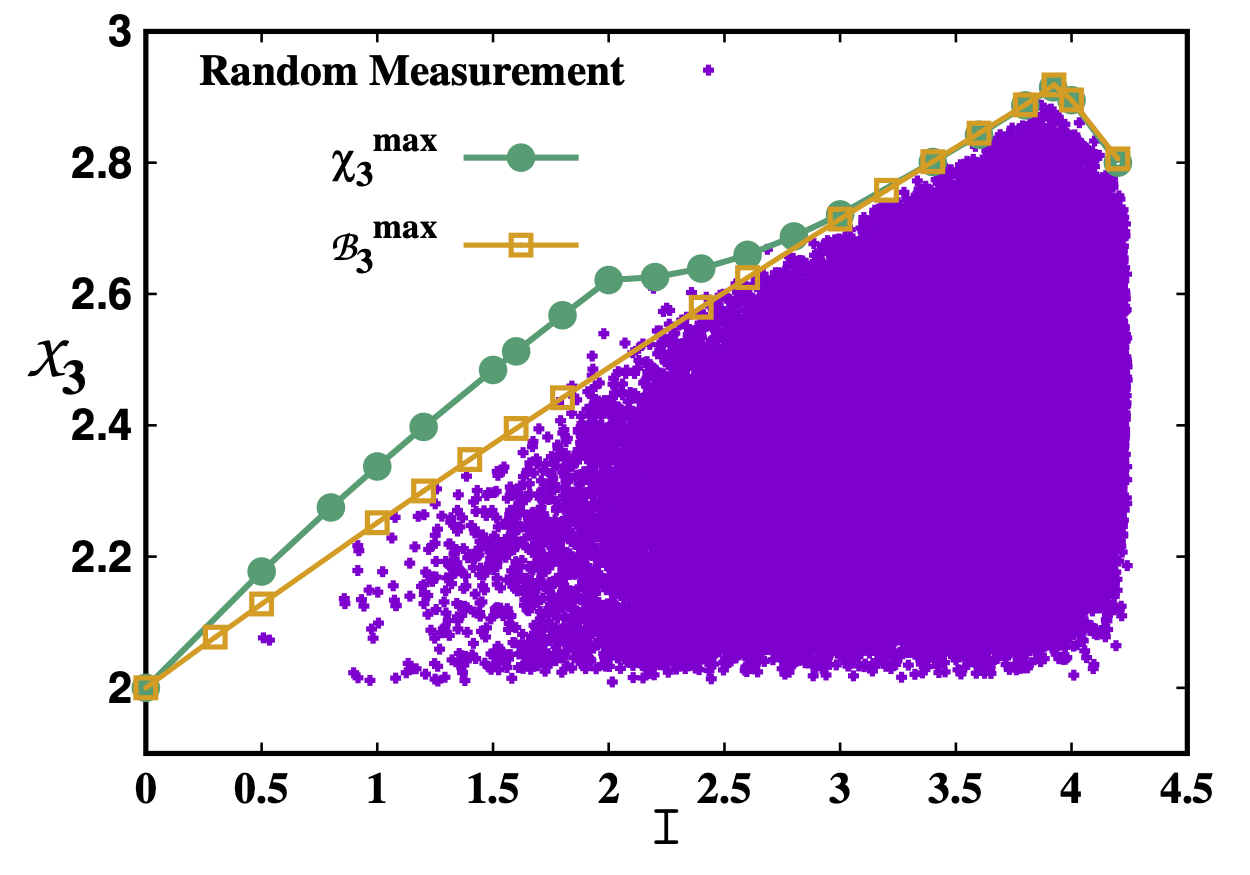}
\caption{(Color online.) \textbf{Amount of CGLMP inequality violation in $d=3$ with incompatibility.}  We Haar uniformly generate $2\times 10^5$ measurement pairs at both Alice's and Bob's side. Violet points correspond to the maximum value of CGLMP violation, $\chi_3$ (vertical axis), given previously mentioned measurement settings while the horizontal axis corresponds to incompatibility, $\mathcal{I}(B^1,B^2)={\tt I}$, of Bob's measurements. For a fixed ${\tt I}$ at Bob's side, \(\chi_3^{\max}\) (see Eq.~\eqref{eq:idmax}), signifies the upper bound of CGLMP violation, optimized over all possible qutrit measurements at both Alice's and Bob's part. \(\mathcal{B}_{3}^{\max}\) (see Eq.~(\ref{eq:idmaxim})) denotes the same, except the optimization is performed over interferometric qutrit measurement settings only. However, interestingly \(\chi_3^{\max}\) and \(\mathcal{B}_{3}^{\max}\) coincide for ${\tt I}\gtrsim 3.2$. Both the axes are dimensionless.} 
\label{fig:incom_viol}
\end{figure}

The above observation raises an important question. Is the nonmonotonic nature of the violation with respect to incompatibility restricted to the case where we consider $\chi_3^{\max}({\tt I})$?  To address this question, we numerically simulate \(2\times 10^5\) random measurements on both Alice and Bob's side as prescribed in Eq.~\eqref{eq:unitary_rand} by Haar uniformly \cite{Zyczkowski1994,deGuise2018} generated unitaries and compute the corresponding maximal violation of the CGLMP inequality in Eq.~\eqref{eq:op_cglmp} which is obtained with the help of Eq.~\eqref{eq:cglmp_operator}. Fig. \ref{fig:incom_viol} depicts the scattered plot of the maximal value of the CGLMP operator, \(\chi_3\) in Eq.~\eqref{eq:op_cglmp} with respect to a fixed \({\tt I}\).
Our main observations from Fig. \ref{fig:incom_viol} are listed below:
\begin{enumerate}
    \item Since by construction, all the random measurement setting pairs are incompatible, the CGLMP inequalities are violated for all those settings, thereby confirming Proposition \ref{prop:sufficient}. 
    
    
    \item From the scattered plot in the \((\chi_3,{\tt I})\)-plane (in  Fig. \ref{fig:incom_viol}), it is clear that \(\chi_3\) is  nonmonotonic with all values of  \({\tt I}\).
\end{enumerate}

Although the results presented are generated by fixing incompatibility in measurement pairs at Bob's side, the qualitative relation between \(\chi_3\) and \(\tt I\) remains the same even when the other side's incompatibility is considered.


\subsection{Nonuniqueness of the maximally violating states}

Let us identify the states that violate CGLMP inequality maximally for a given amount of incompatibility in Alice or Bob's side.  In particular, we are interested to identify state $\rho_{AB}^*$ that maximizes $\chi_d^{\max}({\tt I})$ in Eq. \eqref{eq:idmax}, i.e,
\begin{eqnarray}
     \chi_d^{\max}({\tt I} ) = \max_{\bm{\eta} \in \Theta_{\mathcal{I}(B_0,B_1)={\tt I}}} \text{Tr}[\mathbf{C}_d(\bm{\eta}) \rho_{AB}^*].
\end{eqnarray}
The state  $\rho_{AB}^*$ corresponds to the achievable upper bounds of the violations of the CGLMP inequality for a fixed incompatibility ${\tt I}$ in the local lab of Bob in Fig.~\ref{fig:incom_viol}. A major motivation for such a query is from a feature of the CHSH inequality, where irrespective of the local incompatibility value, the maximally violating state is the maximally entangled two-qubit state.

A detailed investigation for $d = 3$ is enumerated below. 
\begin{enumerate}
\item {\it Nonuniqueness of maximally violating states (MVS).}  Firstly, the maximally violating state is pure, $\rho_{AB}^* = \ketbra{\Psi_{mv}^\gamma}{\Psi_{mv}^\gamma}$. It demonstrates interesting characteristics, which are absent in two-dimensional systems.  In particular,  the states that achieve the upper bound of CGLMP violation in two qutrits for a fixed incompatibility \textcolor{black}{$\mathcal{I}(B^1, B^2)=\tt{I}$ at Bob's side} are local unitarily equivalent to $\ket{\Psi_{mv}^\gamma}$, where    
\begin{eqnarray}
        \ket{\Psi_{mv}^\gamma}=\frac{1}{\sqrt{2+\gamma^2}}(\ket{00}+\gamma\ket{11}+\ket{22}).
    \end{eqnarray}
Here after fitting, we find
    \begin{eqnarray}
        \gamma=\begin{cases}
            0 & \text{for }{\tt{I}}\leq 2,\\\nonumber
            1.541\ln{\tt{I}}-1.084 & \text{for }2\le{\tt{I}} \leq 3.4,\\\nonumber
            -0.050{\tt I} + 0.980 & \text{for }3.4\le{\tt{I}} \leq 3\sqrt{2},
        \end{cases}
    \end{eqnarray} 
    with the coefficient of determination given by $R^2=0.99$. Since the choice of $\gamma$ depends on the local incompatibility ${\tt I}$, \textit{the maximally violating states for two qutrits are non-unique which is, in sharp contrast, to two-qubit scenario}.

 \item {\it MVS in qubits vs qutrits.}   Interestingly, up to ${\tt{I}}=2$, the maximally violating state is local unitarily equivalent to $\ket{\Psi^0_{mv}}=\frac{1}{\sqrt{2}}(\ket{00}+\ket{22})$, the maximally entangled state in the two dimensional subspace. Note that, ${\tt{I}}=2$ is the maximal value of incompatibility which can be reached by two-outcome projective measurements. 
  It possibly indicates that the feature of the CHSH inequality is emerging via the uniqueness of the maximally entangled state in the $\{\ket{00},\ket{22}\}$ subspace. The non-uniqueness of the maximally violating state emerges for ${\tt I}>2$ which can only be achieved by qutrit measurements.

\begin{figure}
\includegraphics [width=\linewidth]{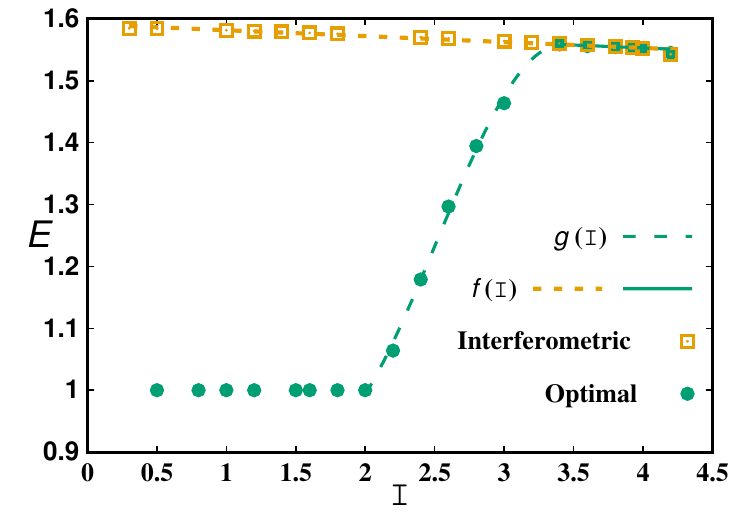}
\caption{(Color online.) \textbf{Entanglement of MVS for a given value of incompatibility.} The entanglement \(E\) (ordinate) of the states, achieving \(\chi_3^{\max}\) (yellow squares) and \(\mathcal{B}_3^{\max}\) (green circles) (see Fig.~\ref{fig:incom_viol}) against a  incompatibility, \({\tt I}\) (abscissa), at Bob's side. Here, \(g({\tt I})=-0.230 {\tt I^3} + 1.735{\tt I^2} - 3.818{\tt I} + 3.532\) (green dashed line) is the fitted curve in the interval of incompatibility, \(2\leq{\tt I}\leq 3.4\) with \(R^2_{g}=0.998\). Similarly, \(f({\tt I})=-0.010 {\tt I} + 1.591\) (yellow dashed line) is the fitted curve in the range, \(0.3\leq {\tt I}\leq 4.2\) with \(R^2_{f}=0.957\). As expected, the entanglement, \(E\) of states, achieving \(\chi_3^{\max}\) matches to those of \(\mathcal{B}_{3}^{\max}\) in the range of incompatibility, \(3.4\lesssim {\tt I} \lesssim 4.2\) where the corresponding upper bounds of nonlocality value coincides (see Fig.~\ref{fig:incom_viol}).
Both the axes are dimensionless.} 
\label{fig:entropy_incm}
\end{figure}

\end{enumerate}
The above analysis emphasizes the dual role of incompatibility and entanglement in the case of maximal violation of CGLMP inequality in qutrits, $d=3$, 
 since for two qubits, a maximally violating state for a given incompatibility is always a maximally entangled state. To visualize the relation of the incompatibility present in measurements leading to MVS and its entanglement content, we observe that \({\tt I}<2\), \(E(\ket{\Psi_{mv}^0})=S(\rho_A)\)=1 where \textcolor{black}{ \(S(\sigma)=-\Tr(\sigma\log_2 \sigma)\) representing von Neumann entropy and  \(\rho_A\) denotes the local density matrix of \(|\Psi_{mv}^\gamma\rangle\)} while for \({\tt I}>2\), \(E(\ket{\Psi_{mv}^\gamma})\) is monotonically increasing with \({\tt I}\) upto ${\tt I} \approx 3.4$ 
(see Fig. \ref{fig:entropy_incm}).  After  ${\tt I} \gtrsim 3.4$, entanglement is approximately saturating.

In the subsequent section, we will demonstrate that this nonmonotonicity between the incompatibility in measurement pairs and the CGLMP violation can also be observed in an experimental setting.


\section{CGLMP violation in an interferometric measurement scheme}
\label{sec:interferometric}
Let us now address the similar three questions answered in the previous section for the general projective measurements
-- $1.$ nonmonotonicity of \( \chi_3\) with \({\tt I}\), $2.$ \(\chi_3^{\max}({\tt I})\) and $3.$ comparison between the quantum states that achieve \(\chi_3^{\max}({\tt I})\) in the case of an experimentally realizable interferometric scheme. In photonic interferometric setup, the maximal violation of the CGLMP inequality (Eq.~\eqref{eq:cglmp}) can be achieved by performing measurements that consist of a tritter \cite{Zukowski1997realizable} (a generalization of $50-50$ beam splitter) and phase shifters \cite{acin2002} which will be summarized below.
Initially, subsystems, A and B, apply local unitary transformation, \( U(\vec{\phi}_a) \) and \( U(\vec{\phi}_b) \), respectively on their individual subsystems. Subsequently, \( A \) performs a discrete Fourier transform, \( U_{FT}=\frac{1}{\sqrt{d}}e^{\big({i(j-1)(k-1)}\frac{2\pi}{d}\big)} \), while \( B \) applies the complex conjugate operation, \( U_{FT}^* \). For the subsystem \( A \), the vector \( \vec{\phi}_a \) is a \( d \)-dimensional vector with components defined as \( \vec{\phi}_a = [\phi_a(0), \phi_a(1), \phi_a(2), \ldots, \phi_a(d-1)] \). Similarly, for the subsystem \( B \), the corresponding vector is denoted by \( \vec{\chi}_b \), where \( \phi_a \) is replaced by \( \chi_b \). The unitary operation on \( A \) is expressed as 
\begin{eqnarray}
    U(\vec{\phi}_a) = \text{diag}\left(e^{i\phi_a(0)}, e^{i\phi_a(1)}, \ldots, e^{i\phi_a(d-1)}\right),
\end{eqnarray}
and an analogous expression holds for \( B \). Both parties then conduct projective measurements in the computational basis, with \( A \) using the projector \( \pi_j^a = \ketbra{j}{j} \) and \( B \) using \( \pi_k^b = \ketbra{k}{k} \). Therefore, the joint probability distribution can be written as
\begin{eqnarray}
  \nonumber  && P(A_j^a,B_k^b) \\
    &=&\Tr[\pi_j^a \otimes \pi_k^b V(\vec{\phi_a})\otimes V(\vec{\chi_b})|\rho_{AB} V(\vec{\phi_a})^{\dag}\otimes V(\vec{\chi_b})^{\dag}]\nonumber\\
    &=& \Tr[V(\vec{\phi_a})^{\dag}\otimes V(\vec{\chi_b})^{\dag}\pi_j^a \otimes \pi_k^b V(\vec{\phi_a})\otimes V(\vec{\chi_b})\rho_{AB} ].\nonumber\\ 
    \label{eq:operator_cglmp_one}
\end{eqnarray}
We can see that, equivalently, the measurement outcomes can now be written as
\begin{eqnarray}
    \nonumber A_j^a &=& V(\vec{\phi}_a)^\dagger\pi_j^aV(\vec{\phi}_a),\\
    B_k^b &=& V(\vec{\chi}_b)^\dagger\pi_k^bV(\vec{\chi}_b),
    \label{eq:opt_meas}
\end{eqnarray}
with \(V({\vec\phi_{a}})=U_{FT}U(\vec\phi_{a})\) and \(V({\vec\chi_{b}})=U_{FT}^{*}U(\vec\chi_{b})\). The maximal violation of the CGLMP inequality can be extracted with measurement settings \cite{acin2002}
\begin{eqnarray}
\phi_{1}(j)=0, \phi_{2}(j)=\frac{j\pi}{d}, \zeta_{1}(j)=\frac{j\pi}{2d}, \zeta_{2}(j)=-\frac{j\pi}{2d}. 
    \label{eq:phases_unitary}
\end{eqnarray}
From Eq.~\eqref{eq:op_cglmp}, it is evident that the maximal value of $\chi_d$ is achieved by taking $\rho_{AB}$ as the eigenstate of $\mathbf{C}_d$ with maximal eigenvalue. As mentioned before, unlike two qubits, for $d\ge 3$, non-maximally entangled states give maximal violation of the CGLMP inequality with the measurement settings described above. Specifically, for $d=3$, the state that violates maximally the CGLMP inequality takes the form as $\ket{\Psi_{mv}^\gamma}=\frac{1}{\sqrt{2+\gamma^2}}(\ket{00}+\gamma\ket{11}+\ket{22})$ with $\gamma=(\sqrt{11}-\sqrt{3})/2$.

The three main questions that we want to address in the interferometric measurement setting are as follows:
\begin{enumerate}
\item Do the phenomena of non-monotonicity of the CGLMP violation with respect to measurement incompatibility persist in this case also?

    \item How well can  $\chi_3^{\max}({\tt I})$ be achieved by optimal measurements in the interferometric setting?

    \item What are the maximal violating states versus measurement incompatibility ${\tt I}$? How does it compare with the ones obtained for general qutrit measurements?
\end{enumerate}

\subsubsection*{Non-monotonicity of the CGLMP violation}


\begin{figure}
\includegraphics [width=\linewidth]{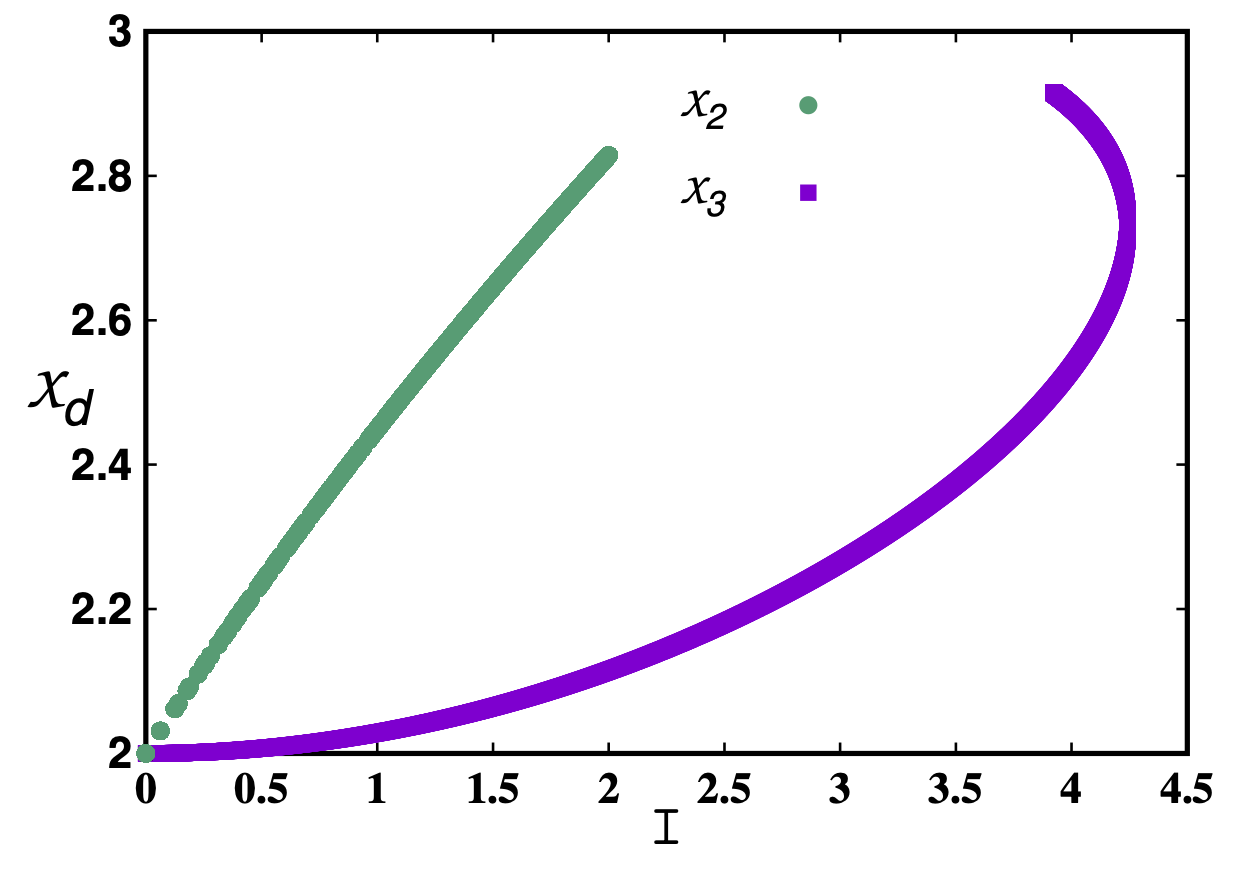}
\caption{(Color online.) \textbf{\(\chi_d\) (vertical axis) vs the incompatibility, \({\tt I}\) (horizontal axis) at Bob's side}. By fixing the interferometric measurements in both sides according to Eqs.~\eqref{eq:opt_meas} and \eqref{eq:phases_unitary} except varying $\zeta_2(1)$ (where $0\leq\zeta_2(1)\leq 2\pi$) of Bob's setting, we calculate \(\chi_3\) (violet squares) for two qutrits whereas, in case of qubits,  CGLMP  which is equivalent to the CHSH inequality, $\chi_2$ (green circles) is calculated for all possible qubit measurements on both Alice's and Bob's side. Notably, \(\chi_2\) varies monotonically with \({\tt I}\), while nonmonotonic behavior is observed for  \(\chi_3\). Both axes in the plot are dimensionless.
} 
\label{fig:cglmp_incom_one_para}
\end{figure}


Let us now concentrate on the qutrit scenario ($d=3$). The measurements on Bob's side are chosen to $\zeta_1(j)=\frac{j\pi}{6}$, and $\zeta_2(j)=-\frac{j\pi}{6}$ for $j=0,1,2$. While, Alice's measurements are $\phi_1(j)=0$ for $j=0,1,2$ and $\phi_2(0)=0$, $\phi_2(1)=\xi\in[0,2\pi]$, $\phi_2(2)=\frac{2\pi}{3}$, i.e, we fixed two settings at Alice's part and change a parameter in one of her settings, \(\xi\), thereby varying incompatibility at Alice's side. Note that these pair of settings do not correspond to MUB bases in \(d=3\) which achieve maximum according to the incompatibility quantifier, \({\tt I}\). 


In contrast to the CHSH scenario,  the maximal violation obtained with measurements in Eqs.~\eqref{eq:opt_meas} and ~\eqref{eq:phases_unitary} again displays nonmonotonic behavior with the variation of incompatibility, {\color{black}confirming the contrasting behavior of qutrits compared to qubits in an experimental situation (see Fig.~\ref{fig:cglmp_incom_one_para}).} 


\subsubsection*{Maximal violation of the CGLMP inequality}

Towards finding the maximal violation of the CGLMP inequality in this interferometric set-up, we modify Eq.~\eqref{eq:idmax} for a fixed local incompatibility value of measurements on any of the two spatially separated parties as 
\begin{eqnarray}
    \mathcal{B}_3^{\max}({\tt I} ) := \max_{\bm{\eta} \in \Theta^{\text{IM}}_{\mathcal{I}(B^1,B^2)={\tt I}}} \max_{\rho_{AB}} \text{Tr}[\mathbf{C}_d(\bm{\eta}) \rho_{AB}], 
    \label{eq:idmaxim}
\end{eqnarray}
where $\Theta^{\text{IM}}$ denotes the set of all interferometric measurements while $\Theta^{\text{IM}}_{\mathcal{I}(B^1,B^2)={\tt I}}$ denotes a subset of $\Theta^{\text{IM}}$ for which the incompatibility of the measurements performed in Bob's lab is fixed to ${\tt I}$. We know \cite{Collins2002bell} that the overall maximal violation of the CGLMP inequality can indeed be achieved with the interferometric measurements, i.e., $$\max_{\tt I} \chi_3^{\max} ({\tt I}) = \max_{\tt I} \mathcal{B}_3^{\max} ({\tt I}) = \sqrt{11/3}\approx 2.9149.$$ 

The key question that we address here is whether this equivalence holds, in general. Our analysis reveals that there exists a critical value of incompatibility, ${\tt I} \approx 3.2$, above which it holds, and $\chi_3^{\max}$ and $\mathcal{B}_3^{\max}$ coincides.
(see Fig.~\ref{fig:incom_viol}). Precisely, we find that when ${\tt I} \leq 3.2$, $\mathcal{B}_3^{\max} ({\tt I})$ is strictly smaller than  $\chi_3^{\max} ({\tt I})$. Therefore, we have
\begin{eqnarray}
    \mathcal{B}_3^{\max} ({\tt I}) = \begin{cases}
       < \chi_3^{\max} ({\tt I}) & \text{ for } {\tt I} \lesssim3.2, \nonumber \\
      = \chi_3^{\max} ({\tt I}) & \text{ for } {\tt I} > 3.2.
    \end{cases}
\end{eqnarray}
Although such deviation of $\mathcal{B}_3^{\max}({\tt I})$ from $ \chi_3^{\max} ({\tt I})$ for small values of incompatibility can be explained as we optimize over a restricted class of measurements, the coinciding feature is surprising.


\subsubsection*{Is the maximally violating state non-unique?}

An interesting finding in the previous section was that the maximally violating state explicitly depends on the incompatibility. This was true under the most general qutrit projective measurements.
Here, we ask whether such a correlation of the maximally violating state exists when we restrict the measurements to ones in the interferometric scheme.

In this scenario, the entanglement of the maximally violating state with the measurement strategy mentioned above shows a qualitatively different response to the local incompatibility compared to the general measurement scenario, as depicted in Fig. \ref{fig:entropy_incm}. In particular, the entanglement entropy of the MVS displays a slow monotonic decrease with $\tt I$, leading to a curious behavior exclusively in the case of interferometric measurement scheme which we can call \textit{more nonlocality with less entanglement.} In this case, too, the maximally violating state reads as (up to local unitaries) $$\ket{\Psi_{mv}^\gamma} = \frac{1}{2+\gamma^2}(\ket{00}+\gamma \ket{11}+\ket{22}),$$
Here $\gamma$ can be written as a function of the local incompatibility $\tt I$ as 
    \begin{eqnarray}
        \gamma= -0.050{\tt I} + 0.980,
    \end{eqnarray}
where $R^2=0.99$.

\subsection{Insufficiency of incompatibility for CGLMP violation}

Although incompatible projective measurements always lead to the violation of the CHSH inequality, it has been shown that incompatible positive operator valued measurements (POVMs) are not sufficient to violate the CHSH inequality \cite{priya_ghosh2023}. Specifically, there exist pairs of POVMs that are not jointly measurable but do not induce a violation of the CHSH inequality. 

A similar question can be raised beyond two qubits by considering CGLMP inequality which remains open till date. We conclusively address this query in Proposition~\ref{prop:insufficiency} in the context of CGLMP inequality for two qutrits. 

\begin{proposition}
    Incompatibility of POVM is not sufficient for violating the local hidden variable model via two input and three output -- based CGLMP inequality.
    \label{prop:insufficiency}
\end{proposition}

\begin{proof}
    Let us assume that Alice's and Bob's measurements are noisy of those considered in Eqs.~\eqref{eq:opt_meas} and ~\eqref{eq:phases_unitary}. Such a consideration may appear naturally in an experimentally realizable set-up due to the non-isolation of the measurement apparatus from the environment. Specifically,  the noisy measurements can be written as 
\begin{eqnarray}
    \nonumber A_j^a &=& \eta V(\vec{\phi}_a)^\dagger\pi_j^aV(\vec{\phi}_a) + (1-\eta)\frac{\mathbb{I}}{3},\\
    B_k^b &=& \eta V(\vec{\zeta}_b)^\dagger\pi_k^bV(\vec{\zeta}_b) + (1-\eta)\frac{\mathbb{I}}{3},
    \label{eq:unsharp_meas}
\end{eqnarray}
where $\eta\in[0,1]$ is the unsharpness parameter in measurements. Under these measurement settings, the maximal value of the CGLMP expression can be obtained as
\begin{eqnarray}
    \chi_3 = (1+\sqrt{11/3})\eta^2,
\end{eqnarray}
where the optimal state that achieves this violation is still $\ket{\Psi_{mv}^\gamma}$ with $\gamma=(\sqrt{11}-\sqrt{3})/2$. It shows that the CGLMP violation is possible for $\eta> \sqrt{\frac{2}{1+\sqrt{11/3}}} \approx 0.83$. However, these measurement settings are not jointly measurable according to the robustness of incompatibility measure defined in Eq.~\eqref{eq:upper_bound_robustness} for $\eta\geq 3/4\approx 0.75$. Therefore, in the region $3/4 < \eta\leq \sqrt{\frac{2}{1+\sqrt{11/3}}}$,  the measurements are incompatible (not jointly measurable), but there does not exist any state which can achieve a violation in the CGLMP inequality. 
\end{proof}


\section{Random Access Codes vs. Incompatibility}
\label{sec:rac_incom}
We now shift our focus to the role of measurement incompatibility in the prepare and measure scenarios, in particular, for the case of random access codes (RACs). Let us begin with a brief discussion of the $2^d\to 1$ random access code (RAC) \cite{Ambainis1999,Ambainis2002,ambainis2009quantumrandomaccesscodes}. In this framework, Alice gets two uniformly random classical dits, $a_1,a_2\in \{0,1,\ldots,d-1\}$ as inputs. Depending on the input, Alice encodes the data into a smaller message of dimension $d$  and communicates it to Bob. Bob receives a uniformly random input $b\in\{1,2\}$, based on which his task is to guess the value of Alice's dit $a_b$. Hence, Bob's output $j$ is the estimated value of $a_b$. The average success probability is defined as
\begin{eqnarray}
    \mathcal{R}_{d}&=&\frac{1}{2d^{2}}\sum_{a_1,a_2=0}^{d-1}\sum_{b=1}^{2} P(j=a_{b}|a_1,a_2b).\nonumber\\
    \label{eq:RAC_dit}
\end{eqnarray}
Classically, Bob can win the game with a maximum average success probability of $\mathcal{R}_d^{C,m}=\frac{1}{2}(1+\frac{1}{d})$. However, in the case of quantum random access code (QRAC), Bob chooses the encoding state to be qudit, $\rho_{a_1a_2}$ of dimension $d$, and the decoding strategy is performed by a measurement on $\rho_{a_1a_2}$. The measurement is described by $B^{b}\equiv\{B^b_{j=a_b}\}_{a_b}$. Therefore, the average success probability in the quantum case can be expressed as 
\begin{eqnarray}
    \mathcal{R}_{d}^Q&=&\frac{1}{2d^{2}}\sum_{a_1,a_2=0}^{d-1}\tr[\rho_{a_1a_2}(B^1_{a_1}+B^2_{a_2})].\nonumber\\
    \label{eq:RAC_qdit}
\end{eqnarray}
The optimal average success probability in the quantum scenario, $\mathcal{R}_d^{Q,m}=\frac{1}{2}\left(1+\frac{1}{\sqrt{d}}\right)>\mathcal{R}_d^{C,m}$,  can be achieved by choosing the $B^1$ and $B^2$ as rank-$1$ projective measurements in mutually unbiased bases and $\rho_{a_1a_2}$ as the eigenstates of $B^1_{a_1}+B^2_{a_2}$, corresponding to the maximum eigenvalue \cite{Farkas_2019}. Therefore, if the encoding state is chosen optimally, we can write \cite{Mordasewicz2022}
\begin{eqnarray}
    \mathcal{R}_d^Q=\frac{1}{2d^2}\sum_{a_1,a_2=0}^{d-1}||B^1_{a_1}+B^2_{a_2}||_{\infty}.
    \label{eq:racsuccessprob}
\end{eqnarray}

\subsection{Inequivalence of higher dimensional RAC and nonlocality}
Let us now investigate the role of measurement incompatibility to obtain the average success probability in QRAC. Moreover, we exhibit how higher dimensional nonlocality captured by the violations of the CGLMP inequality possesses a qualitatively different response to measurement incompatibility in contrast to the average success probability obtained in QRAC. Before presenting the result for arbitrary dimension, $d$, we demonstrate that in the case of $d=2$, the average success probability of the QRAC game can be explicitly cast in terms of measurement incompatibility. 
{\color{black}
First, note that both $\mathcal{R}_2^Q$ and $\mathcal{I}_p$ can be written as functions of a single parameter $x$, 
\begin{eqnarray}
  \mathcal{R}^Q_2 &=& \frac14 (2+ x + \sqrt{1-x^2})  \\
   \mathcal{I}_p\equiv\mathcal{I}_p(B^1,B^2) &=& 2^{2 + \frac1p} x\sqrt{1-x^2},
\end{eqnarray}
where for two arbitrary projective measurements $B^1 = \{B^1_1,B^1_2\}$ and $B^2 = \{B^2_1,B^2_2\}$, $x = |\langle B^1_1|B^2_1\rangle|$. Combining the above two expressions, we get
\begin{eqnarray}
    \mathcal{R}^Q_2 = \frac12 + \frac14 \sqrt{1+ \frac{\mathcal{I}_p}{2^{1+\frac1p}}}.
    \label{eq:r2q}
\end{eqnarray}
The details of the derivations are provided in the Appendix. \ref{app:rac221}. Another interesting result that directly follows from Eqs. \eqref{eq:chi2} and  \eqref{eq:r2q} is
\begin{eqnarray}
    \mathcal{R}_2^Q = \frac{1 + \chi_2}{2},
\end{eqnarray}
which explicitly demonstrates the isomorphism between the average QRAC success probability and the CHSH violation.
}

\begin{figure}
\includegraphics [width=\linewidth]{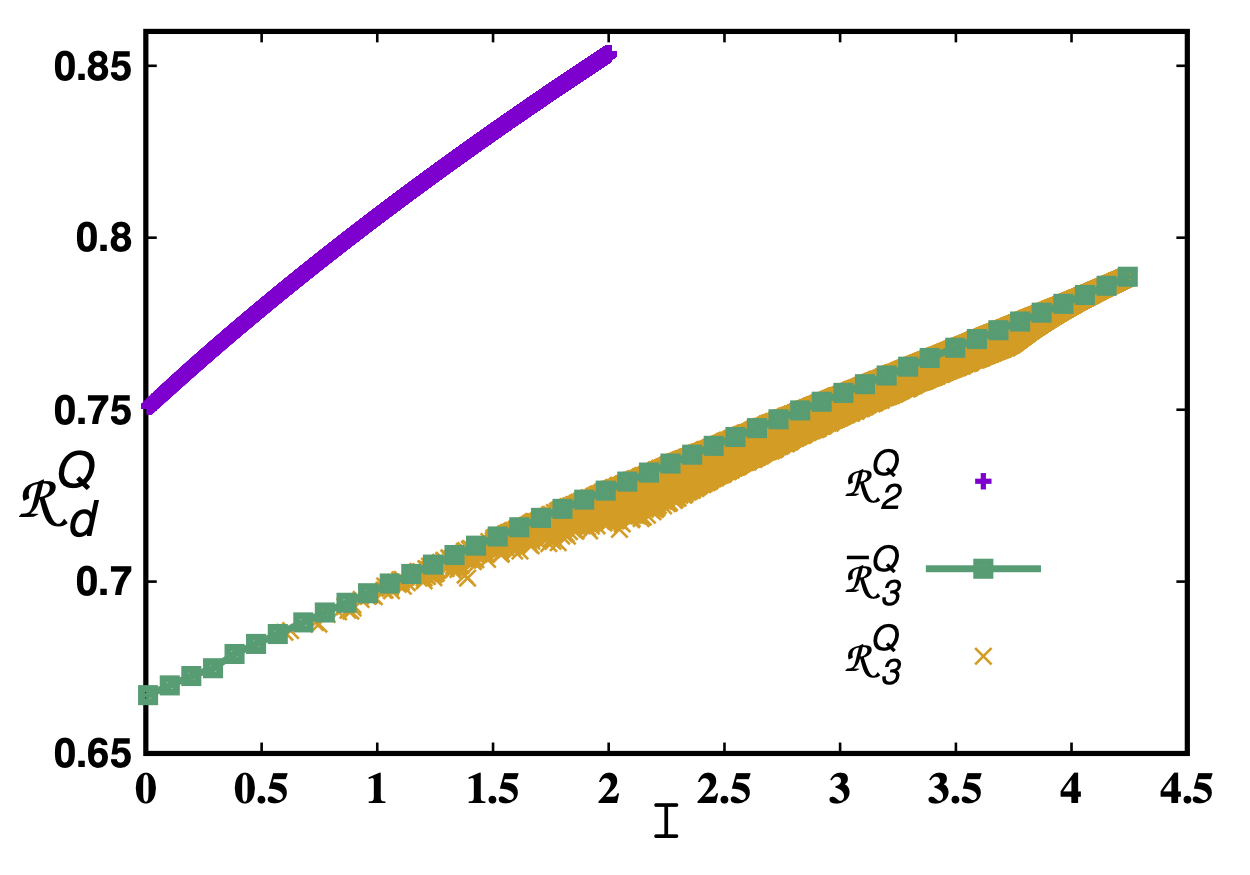}
\caption{(Color online.) \textbf{Maximum average success probability of $2^d\to 1$ QRAC game.} \(\mathcal R_d^{Q}\) (ordinate) against incompatibility, \({\tt I}\) (abscissa) of measurements performed by Bob.  Note that \(\mathcal{R}_2^{Q}\)(violet points) varies monotonically with Bob's incompatibility \({\tt I}\), whereas \(\mathcal{R}_3^{Q}\)(yellow crosses) displays non-monotonic behavior. Interestingly, upper bound of average success probability, \(\mathcal{\bar R}_3^{Q}\)(green squares) remains monotonic with respect to ${\tt I}$. Here, both horizontal and vertical axes are dimensionless.} 
\label{fig:rac_d2_d3}
\end{figure}
For a general dimension $d$, the average success probability in the QRAC game and the measurement incompatibility can be formally expressed as
\begin{eqnarray}
    \mathcal{R}_d^Q &=& \frac12 + \frac{1}{2d^2}\sum_{i,j=0}^{d-1} x_{ij}, \nonumber \\
    \mathcal{I}_p &=&  2^{\frac1p}\sum_{i,j=0}^{d-1} x_{ij}\sqrt{1 - x_{ij}^2},
\end{eqnarray}
where $x_{ij} = |\langle B^1_i|B^2_j\rangle|$. However, unfortunately, unlike in the case of $d = 2$, expressing $\mathcal{R}_d^Q$ as an explicit function of $\mathcal{I}_p$ is not straightforward. Even for $d = 3$, it becomes cumbersome.  Nevertheless, we try to establish the relationship between the two for  $d > 2$ using numerical techniques.

Motivated by our analysis of the CGLMP inequality, we consider the maximal average success probability of the QRAC game when the decoding measurements possess fixed incompatibility. Mathematically, it can be written as
\begin{eqnarray}
    \overline{\mathcal{R}}_d^{Q}({\tt I}) := \frac{1}{2d^2}~~\max_{\mathcal{I}(B^1,B^2)={\tt I}} \sum_{a_1,a_2=0}^{d-1}||B^1_{a_1}+B^2_{a_2}||_{\infty}.
\end{eqnarray}
A thorough numerical investigation for $d=3$ reveals that $\overline{\mathcal{R}}_3^Q({\tt I})$ is indeed a monotonic function of the incompatibility ${\tt I}$, as shown in Fig.~\ref{fig:rac_d2_d3}. This clearly illustrates the qualitatively different response to measurement incompatibility of the average success probability in the QRAC game and the CGLMP violation reported in the previous section. 

The above investigation raises an important question -- \textit{Is $\mathcal{R}_3^Q$ monotonic with respect to measurement incompatibility?} We find the answer to be negative. Although $\overline{\mathcal{R}}_3^Q({\tt I})$ is monotonic with the measurement incompatibility, when we randomly (Haar uniformly) generate the measurements $B^1$ and $B^2$, we find instances where $\mathcal{R}_3^Q$ is more for less incompatible measurements, as clearly visible from scattered nature of $\mathcal{R}_3^Q$ in Fig.~\ref{fig:rac_d2_d3}.  

We continue to highlight the inequivalence of higher dimensional QRAC and CGLMP violations by taking the maximal CGLMP violating measurement settings at Bob's side (see Eq.~\eqref{eq:opt_meas}), i.e., $B^b_{a_b}=V(\vec{\chi}_b)^\dagger\pi_{a_b}^bV(\vec{\chi}_b)$ as the decoding strategy of QRAC. Now, if the measurement becomes noisy, which can be expressed as
\begin{eqnarray}
    B^b_{a_b} \to \eta B^b_{a_b} + \frac{1-\eta}{d} \mathbb{I},
    \label{eq:noisy_B}
\end{eqnarray}
with $\eta$ being the unsharp parameter as in Eq.~\eqref{eq:unsharp_meas}. Let us define two critical values of $\eta$ that guarantee non-classical features in case of the CGLMP violation and the QRAC game. Specifically, 
\begin{enumerate}
    \item $\eta_{d,r}$ defines the critical value of $\eta$ such that for $\eta>\eta_{d,r}$, we obtain a quantum advantage in the QRAC game, and
    \item $\eta_{d,c}$ denotes the critical value of noise such that  $\eta>\eta_{d,c}$ leads to  violation of the CGLMP inequality.
\end{enumerate}
In the case of the CGLMP scenario, Alice chose the optimal measurement setting corresponding to Eqs.~\eqref{eq:opt_meas},  \eqref{eq:phases_unitary}, and the shared entangled state is chosen to be optimal. 
\begin{figure}
\includegraphics [width=\linewidth]{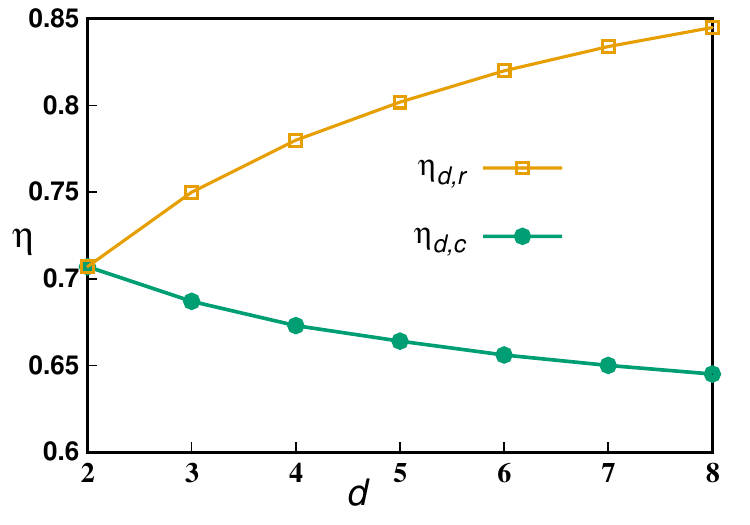}
\caption{(Color online.) \textbf{Critical noise parameter, \(\eta\) (ordinate), responsible for quantum advantage in  QRAC and violation for CGLMP with respect to different dimensions, \(d\) (abscissa)}. In both cases, Bob performs measurement according to Eqs.~\eqref{eq:opt_meas} and \eqref{eq:phases_unitary} which are optimal for extracting nonlocality in the CGLMP scenario and Alice's set of optimal measurements in the CGLMP scenario is given by the same set of equations. Here, \(\eta_{d,r}\) (yellow squares) represents the critical value of \(\eta\) for quantum advantage in the QRAC game while \(\eta_{d,c}\) (green circles) corresponds to the critical value for the same in the context of CGLMP inequality. Here, both the axes are dimensionless.} 
\label{fig:dim_eta_cglmp_rac}
\end{figure}
We find that for $d\geq 3$ $\eta_{d,c}\neq\eta_{d,r}$ (see Table~\ref{tab:eta_r_c}), thereby, again establishing the inequivalence between the two. Interestingly, with increasing $d$, $\eta_{d,r}$ monotonically increases, see Table. \ref{tab:eta_r_c} and Fig.~\ref{fig:dim_eta_cglmp_rac}, indicating a decrease in robustness for noise to the measurements with higher $d$ with respect to the average success probability in the QRAC game. On the contrary, $\eta_{d,c}$ monotonically decreases with the noise in the measurements pointing towards enhanced robustness of the CGLMP violation to measurement noise.

\begin{table}[h!]
\centering
\renewcommand{\arraystretch}{1.5}
\begin{tabular}{|c|c|c|c|}
\hline
\textbf{Dimension} $(d)$ & ~~~~$\eta_{d,c}$~~~~ & ~~~~$\eta_{d,r}$~~~~ \\
\hline

2 & 0.707 & 0.707  \\ \hline
3 & 0.687 & 0.75   \\ \hline
4 & 0.673 & 0.78   \\ \hline
5 & 0.664 & 0.802  \\ \hline
6 & 0.656 & 0.82   \\ \hline
7 & 0.65  & 0.834  \\ \hline
8 & 0.645 & 0.845  \\ \hline
\end{tabular}
\caption{Data points corresponding to Fig.~\ref{fig:dim_eta_cglmp_rac}. The amount of white noise added to the measurement that can sustain quantum advantage in the RAC game $1-\eta_{d,r}$, and the violation of the CGLMP inequality $1 - \eta_{d,c}$ is given for increasing dimensions $(2 \leq d \leq 8)$.
}
\label{tab:eta_r_c}
\end{table}

\subsubsection{Construction  of strategies with $\eta_{d,r}=\eta_{d,c}$}


Having discussed the divergent behavior of QRAC and the CGLMP scenario, let us now address a question which can put them on same footing -- \textit{Given a fixed value of incompatibility, does there exist any strategy in the CGLMP scenario such that violation of the CGLMP inequality implies a quantum advantage in QRAC game and vice versa, i.e., $\eta_{d,r}=\eta_{d,c}$ for $d \geq 2$}?

To answer this query, we Haar uniformly generate measurements $A^1$ and $A^2$ at Alice's side. We fix the measurements at Bob's side to be $B^1=A^1$ and $B^2=A^2$. For noisy measurements on Bob's side as shown in Eq.~\eqref{eq:noisy_B}, the shared state is optimally chosen which remains the same for all values of $\eta$.

In the qubit case, out of $2\times 10^5$ Haar uniformly generated qubit measurement pairs, there does not exist any measurement scheme that can achieve the above mentioned goal except when the measurement set is maximally incompatible,  achieved by MUBs, i.e., $\mathcal{I}(B^1,B^2)=2$. In fact, we also simulate $10^5$ qubit MUB pairs, all of which exhibit identical robustness for CGLMP and QRAC scenarios aligning with our interest. Note that, due to the monotonically increasing nature of quantum advantage with incompatibility in both QRAC and CGLMP frameworks, these projective MUB settings provide  $\mathcal{R}^{\mathcal{Q}}_2 \approx 0.8534 $ and $\chi_2^{\max}\approx2.8284$. 

\begin{figure}
\includegraphics [width=\linewidth]{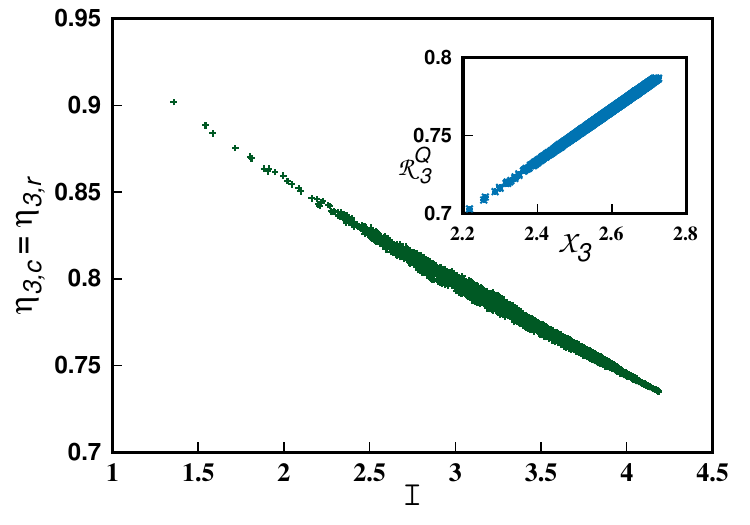}
\caption{(Color online.) The range of \(\eta_{3,c}=\eta_{3,r}\) (ordinate) (quantum advantage in QRAC and violations of CGLMP inequality occur simultaneously) against incompatibility, \({\tt I}\) (abscissa) of Bob's side. Here we study the situation where Alice performs the same measurement strategy as Bob. The inset plot shows the Value of QRAC, \(\mathcal{R}_{3}^{Q}\) with the CGLMP operator, \(\chi_3\) for  measurement settings with $\eta=1$. Here, both the axes are dimensionless.} 
\label{fig:rac_cglmp_d3}
\end{figure}

However, this singular behavior does not exist in a higher dimension. We show that, in the case of $d=3$, there exist measurement schemes for almost the whole range of incompatibility, $1.36\lesssim\mathcal{I}(B^1,B^2)\lesssim 4.19$ at Bob's side such that  $\eta_{3,r}=\eta_{3,c}$ can be found as demonstrated in Fig.~\ref{fig:rac_cglmp_d3} in which $2\times 10^5$ measurement strategies at Bob's side are Haar uniformly generated. Interestingly, we find a finite number of measurement settings $(\approx 2 \%)$ out of the total simulated measurements that satisfies $\eta_{3,r}=\eta_{3,c}$. The characterization of these measurements in qutrit scenario would be an interesting task going forward. Let us take any two these measurement strategies with $\eta=1$, i.e., $\mathbb{B}=\{B^1,B^2\}$ and $\mathbb{B}'=\{B'^1,B'^2\}$. The average success probability in QRAC game corresponding to $\mathbb{B}$ and $\mathbb{B'}$ are denoted by $\mathcal{R}^{\mathcal{Q}}_3(\mathbb{B})$ and $\mathcal{R}^{\mathcal{Q}}_3(\mathbb{B'})$ respectively. Similarly, in the CGLMP scenario, we have $\chi_3(\mathbb{B})$ and $\chi_3(\mathbb{B'})$. The inset of Fig.~\ref{fig:rac_cglmp_d3} shows that for any two pair of  measurement settings  with $\mathcal{R}^{\mathcal{Q}}_3(\mathbb{B})\geq\mathcal{R}^{\mathcal{Q}}_3(\mathbb{B'})$ implies $\chi_3(\mathbb{B})\geq\chi_3(\mathbb{B'})$ and vice versa.

In this way, our analysis reflects that due to equal robustness, given one of these measurement settings if we get a violation of CGLMP inequality, then it certifies quantum advantage in the QRAC game and vice versa. 

\section{Conclusion}
\label{sec:conclusion}
In recent years, a myriad of works have studied the role of incompatibility in the $2 (party) -2 (settings) -2 (outcome) $ Bell scenario, however, the possibility remains open in the higher-dimensional Bell scenario. This work strives to establish a connection between incompatibility and Bell nonlocality beyond the standard $2-2-2$ platform to the $d$ outcome $2-2-d$ scenario by considering the CGLMP inequality. We explicitly demonstrated that the incompatibility in projective measurements is both necessary and sufficient for obtaining the violation of CGLMP inequality in the $d = 3$ scenario and argued for higher dimensions. Specifically, the amount of CGLMP violation with the incompatibility quantifier exhibits contrasting behavior in higher dimensions compared to the $d=2$ picture. In particular, while for $d=2$, there is a monotonic behavior with measurement incompatibility, violation of the CGLMP inequality displays nonmonotonicity with incompatibility in measurements. Moreover, the maximum value of CGLMP violation for a given value of measurement incompatibility at any one side also exhibits this nonmonotonic behavior. This encapsulates one of our central results: \textit{more nonlocality with less incompatibility.} 

Interestingly, the maximally violating state turns out to be nonunique and explicitly depends on the incompatibility of measurements used to extract the Bell statistics. In particular, we identify the maximally violating states to be local unitarily equivalent to a single parameter family of states, expressed as a function of incompatibility. It is in sharp contrast with the $d=2$ Bell scenario, where the maximally entangled state is maximally violating for any value of incompatibility. Additionally, we found that the feature of more nonlocality with less incompatibility exists in the photonic interferometric measurement scheme, thereby raising the possibility of demonstrating this feature in experiments. Deviating from the projective measurement scenario, we proved the insufficiency in the incompatibility of general positive operator-valued measurement for the CGLMP violation with an explicit example.

We also explored the comparison between the Bell scenario and the quantum random access code (QRAC) in the prepare-measure scenario. In the qubit case, i.e., $2^2\to 1$ QRAC, the average success probability is a monotonic function of incompatibility. However, $2^3\to 1$ QRAC shows nonmonotonic characteristics with incompatibility like the qutrit CGLMP scenario. On the other hand, for a fixed incompatibility, the maximum average success probability of QRAC is a monotonic function of incompatibility, which is, in contrast, to the violation of CGLMP inequality for two qutrits. 

Towards connecting the robustness of CGLMP violation with the average success probability in QRAC, we devise noisy measurement schemes that demonstrate nonlocality and quantum advantage in QRAC precisely with equal noise content of the measurements. Interestingly, compared to the qubit scenario where only mutually unbiased bases (MUBs) exhibit such phenomena, we show that in the qutrit case, there exist such measurements for the whole range of incompatibility.

Our work explicitly demonstrates the contrasting role of incompatibility concerning the violation of higher dimensional Bell inequalities and prepare-measure scenarios compared to the two-qubit cases. Our investigations open up new avenues for research into the joint measurability of projective and general positive operator valued measurements (POVMs) in Bell and prepare-measure situations in higher dimensions.

\acknowledgements
SR acknowledges useful discussions with Tamal Guha.  We acknowledge the use of \href{https://github.com/titaschanda/QIClib}{QIClib} -- a modern C++ library for general purpose quantum information processing and quantum computing (\url{https://titaschanda.github.io/QIClib}) and cluster computing facility at Harish-Chandra Research Institute. PH acknowledges ``INFOSYS
scholarship for senior students".

\bibliography{bib.bib}

\begin{thebibliography}{67}%
\makeatletter
\providecommand \@ifxundefined [1]{%
 \@ifx{#1\undefined}
}%
\providecommand \@ifnum [1]{%
 \ifnum #1\expandafter \@firstoftwo
 \else \expandafter \@secondoftwo
 \fi
}%
\providecommand \@ifx [1]{%
 \ifx #1\expandafter \@firstoftwo
 \else \expandafter \@secondoftwo
 \fi
}%
\providecommand \natexlab [1]{#1}%
\providecommand \enquote  [1]{``#1''}%
\providecommand \bibnamefont  [1]{#1}%
\providecommand \bibfnamefont [1]{#1}%
\providecommand \citenamefont [1]{#1}%
\providecommand \href@noop [0]{\@secondoftwo}%
\providecommand \href [0]{\begingroup \@sanitize@url \@href}%
\providecommand \@href[1]{\@@startlink{#1}\@@href}%
\providecommand \@@href[1]{\endgroup#1\@@endlink}%
\providecommand \@sanitize@url [0]{\catcode `\\12\catcode `\$12\catcode `\&12\catcode `\#12\catcode `\^12\catcode `\_12\catcode `\%12\relax}%
\providecommand \@@startlink[1]{}%
\providecommand \@@endlink[0]{}%
\providecommand \url  [0]{\begingroup\@sanitize@url \@url }%
\providecommand \@url [1]{\endgroup\@href {#1}{\urlprefix }}%
\providecommand \urlprefix  [0]{URL }%
\providecommand \Eprint [0]{\href }%
\providecommand \doibase [0]{http://dx.doi.org/}%
\providecommand \selectlanguage [0]{\@gobble}%
\providecommand \bibinfo  [0]{\@secondoftwo}%
\providecommand \bibfield  [0]{\@secondoftwo}%
\providecommand \translation [1]{[#1]}%
\providecommand \BibitemOpen [0]{}%
\providecommand \bibitemStop [0]{}%
\providecommand \bibitemNoStop [0]{.\EOS\space}%
\providecommand \EOS [0]{\spacefactor3000\relax}%
\providecommand \BibitemShut  [1]{\csname bibitem#1\endcsname}%
\let\auto@bib@innerbib\@empty
\bibitem [{\citenamefont {G\"uhne}\ \emph {et~al.}(2023)\citenamefont {G\"uhne}, \citenamefont {Haapasalo}, \citenamefont {Kraft}, \citenamefont {Pellonp\"a\"a},\ and\ \citenamefont {Uola}}]{gunhe2023incompatibility}%
  \BibitemOpen
  \bibfield  {author} {\bibinfo {author} {\bibfnamefont {O.}~\bibnamefont {G\"uhne}}, \bibinfo {author} {\bibfnamefont {E.}~\bibnamefont {Haapasalo}}, \bibinfo {author} {\bibfnamefont {T.}~\bibnamefont {Kraft}}, \bibinfo {author} {\bibfnamefont {J.-P.}\ \bibnamefont {Pellonp\"a\"a}}, \ and\ \bibinfo {author} {\bibfnamefont {R.}~\bibnamefont {Uola}},\ }\href {\doibase 10.1103/RevModPhys.95.011003} {\bibfield  {journal} {\bibinfo  {journal} {Rev. Mod. Phys.}\ }\textbf {\bibinfo {volume} {95}},\ \bibinfo {pages} {011003} (\bibinfo {year} {2023})}\BibitemShut {NoStop}%
\bibitem [{\citenamefont {Horodecki}\ \emph {et~al.}(2009)\citenamefont {Horodecki}, \citenamefont {Horodecki}, \citenamefont {Horodecki},\ and\ \citenamefont {Horodecki}}]{hhhh}%
  \BibitemOpen
  \bibfield  {author} {\bibinfo {author} {\bibfnamefont {R.}~\bibnamefont {Horodecki}}, \bibinfo {author} {\bibfnamefont {P.}~\bibnamefont {Horodecki}}, \bibinfo {author} {\bibfnamefont {M.}~\bibnamefont {Horodecki}}, \ and\ \bibinfo {author} {\bibfnamefont {K.}~\bibnamefont {Horodecki}},\ }\href {\doibase 10.1103/RevModPhys.81.865} {\bibfield  {journal} {\bibinfo  {journal} {Rev. Mod. Phys.}\ }\textbf {\bibinfo {volume} {81}},\ \bibinfo {pages} {865} (\bibinfo {year} {2009})}\BibitemShut {NoStop}%
\bibitem [{\citenamefont {Bell}(1964)}]{Bell1964}%
  \BibitemOpen
  \bibfield  {author} {\bibinfo {author} {\bibfnamefont {J.~S.}\ \bibnamefont {Bell}},\ }\href {\doibase 10.1103/PhysicsPhysiqueFizika.1.195} {\bibfield  {journal} {\bibinfo  {journal} {Physics Physique Fizika}\ }\textbf {\bibinfo {volume} {1}},\ \bibinfo {pages} {195} (\bibinfo {year} {1964})}\BibitemShut {NoStop}%
\bibitem [{\citenamefont {Brunner}\ \emph {et~al.}(2014)\citenamefont {Brunner}, \citenamefont {Cavalcanti}, \citenamefont {Pironio}, \citenamefont {Scarani},\ and\ \citenamefont {Wehner}}]{Bruner2014bell}%
  \BibitemOpen
  \bibfield  {author} {\bibinfo {author} {\bibfnamefont {N.}~\bibnamefont {Brunner}}, \bibinfo {author} {\bibfnamefont {D.}~\bibnamefont {Cavalcanti}}, \bibinfo {author} {\bibfnamefont {S.}~\bibnamefont {Pironio}}, \bibinfo {author} {\bibfnamefont {V.}~\bibnamefont {Scarani}}, \ and\ \bibinfo {author} {\bibfnamefont {S.}~\bibnamefont {Wehner}},\ }\href {\doibase 10.1103/RevModPhys.86.419} {\bibfield  {journal} {\bibinfo  {journal} {Rev. Mod. Phys.}\ }\textbf {\bibinfo {volume} {86}},\ \bibinfo {pages} {419} (\bibinfo {year} {2014})}\BibitemShut {NoStop}%
\bibitem [{\citenamefont {Barnett}\ and\ \citenamefont {Croke}(2009)}]{Barnett2009}%
  \BibitemOpen
  \bibfield  {author} {\bibinfo {author} {\bibfnamefont {S.~M.}\ \bibnamefont {Barnett}}\ and\ \bibinfo {author} {\bibfnamefont {S.}~\bibnamefont {Croke}},\ }\href {\doibase 10.1364/aop.1.000238} {\bibfield  {journal} {\bibinfo  {journal} {Advances in Optics and Photonics}\ }\textbf {\bibinfo {volume} {1}},\ \bibinfo {pages} {238} (\bibinfo {year} {2009})}\BibitemShut {NoStop}%
\bibitem [{\citenamefont {Bae}\ and\ \citenamefont {Kwek}(2015)}]{Bae2015}%
  \BibitemOpen
  \bibfield  {author} {\bibinfo {author} {\bibfnamefont {J.}~\bibnamefont {Bae}}\ and\ \bibinfo {author} {\bibfnamefont {L.-C.}\ \bibnamefont {Kwek}},\ }\href {\doibase 10.1088/1751-8113/48/8/083001} {\bibfield  {journal} {\bibinfo  {journal} {Journal of Physics A: Mathematical and Theoretical}\ }\textbf {\bibinfo {volume} {48}},\ \bibinfo {pages} {083001} (\bibinfo {year} {2015})}\BibitemShut {NoStop}%
\bibitem [{\citenamefont {Skrzypczyk}\ \emph {et~al.}(2019)\citenamefont {Skrzypczyk}, \citenamefont {\ifmmode \check{S}\else \v{S}\fi{}upi\ifmmode~\acute{c}\else \'{c}\fi{}},\ and\ \citenamefont {Cavalcanti}}]{sd1}%
  \BibitemOpen
  \bibfield  {author} {\bibinfo {author} {\bibfnamefont {P.}~\bibnamefont {Skrzypczyk}}, \bibinfo {author} {\bibfnamefont {I.}~\bibnamefont {\ifmmode \check{S}\else \v{S}\fi{}upi\ifmmode~\acute{c}\else \'{c}\fi{}}}, \ and\ \bibinfo {author} {\bibfnamefont {D.}~\bibnamefont {Cavalcanti}},\ }\href {\doibase 10.1103/PhysRevLett.122.130403} {\bibfield  {journal} {\bibinfo  {journal} {Phys. Rev. Lett.}\ }\textbf {\bibinfo {volume} {122}},\ \bibinfo {pages} {130403} (\bibinfo {year} {2019})}\BibitemShut {NoStop}%
\bibitem [{\citenamefont {Carmeli}\ \emph {et~al.}(2018)\citenamefont {Carmeli}, \citenamefont {Heinosaari},\ and\ \citenamefont {Toigo}}]{sd2}%
  \BibitemOpen
  \bibfield  {author} {\bibinfo {author} {\bibfnamefont {C.}~\bibnamefont {Carmeli}}, \bibinfo {author} {\bibfnamefont {T.}~\bibnamefont {Heinosaari}}, \ and\ \bibinfo {author} {\bibfnamefont {A.}~\bibnamefont {Toigo}},\ }\href {\doibase 10.1103/PhysRevA.98.012126} {\bibfield  {journal} {\bibinfo  {journal} {Phys. Rev. A}\ }\textbf {\bibinfo {volume} {98}},\ \bibinfo {pages} {012126} (\bibinfo {year} {2018})}\BibitemShut {NoStop}%
\bibitem [{\citenamefont {Guerini}\ \emph {et~al.}(2019)\citenamefont {Guerini}, \citenamefont {Quintino},\ and\ \citenamefont {Aolita}}]{sd3}%
  \BibitemOpen
  \bibfield  {author} {\bibinfo {author} {\bibfnamefont {L.}~\bibnamefont {Guerini}}, \bibinfo {author} {\bibfnamefont {M.~T.}\ \bibnamefont {Quintino}}, \ and\ \bibinfo {author} {\bibfnamefont {L.}~\bibnamefont {Aolita}},\ }\href {\doibase 10.1103/PhysRevA.100.042308} {\bibfield  {journal} {\bibinfo  {journal} {Phys. Rev. A}\ }\textbf {\bibinfo {volume} {100}},\ \bibinfo {pages} {042308} (\bibinfo {year} {2019})}\BibitemShut {NoStop}%
\bibitem [{\citenamefont {Ambainis}\ \emph {et~al.}(2002)\citenamefont {Ambainis}, \citenamefont {Nayak}, \citenamefont {Ta-Shma},\ and\ \citenamefont {Vazirani}}]{Ambainis2002}%
  \BibitemOpen
  \bibfield  {author} {\bibinfo {author} {\bibfnamefont {A.}~\bibnamefont {Ambainis}}, \bibinfo {author} {\bibfnamefont {A.}~\bibnamefont {Nayak}}, \bibinfo {author} {\bibfnamefont {A.}~\bibnamefont {Ta-Shma}}, \ and\ \bibinfo {author} {\bibfnamefont {U.}~\bibnamefont {Vazirani}},\ }\href {\doibase 10.1145/581771.581773} {\bibfield  {journal} {\bibinfo  {journal} {Journal of the ACM}\ }\textbf {\bibinfo {volume} {49}},\ \bibinfo {pages} {496–511} (\bibinfo {year} {2002})}\BibitemShut {NoStop}%
\bibitem [{\citenamefont {Grudka}\ \emph {et~al.}(2014)\citenamefont {Grudka}, \citenamefont {Horodecki}, \citenamefont {Horodecki}, \citenamefont {K\l{}obus},\ and\ \citenamefont {Paw\l{}owski}}]{rac2}%
  \BibitemOpen
  \bibfield  {author} {\bibinfo {author} {\bibfnamefont {A.}~\bibnamefont {Grudka}}, \bibinfo {author} {\bibfnamefont {K.}~\bibnamefont {Horodecki}}, \bibinfo {author} {\bibfnamefont {M.}~\bibnamefont {Horodecki}}, \bibinfo {author} {\bibfnamefont {W.}~\bibnamefont {K\l{}obus}}, \ and\ \bibinfo {author} {\bibfnamefont {M.}~\bibnamefont {Paw\l{}owski}},\ }\href {\doibase 10.1103/PhysRevLett.113.100401} {\bibfield  {journal} {\bibinfo  {journal} {Phys. Rev. Lett.}\ }\textbf {\bibinfo {volume} {113}},\ \bibinfo {pages} {100401} (\bibinfo {year} {2014})}\BibitemShut {NoStop}%
\bibitem [{\citenamefont {Tavakoli}\ \emph {et~al.}(2015)\citenamefont {Tavakoli}, \citenamefont {Hameedi}, \citenamefont {Marques},\ and\ \citenamefont {Bourennane}}]{rac3}%
  \BibitemOpen
  \bibfield  {author} {\bibinfo {author} {\bibfnamefont {A.}~\bibnamefont {Tavakoli}}, \bibinfo {author} {\bibfnamefont {A.}~\bibnamefont {Hameedi}}, \bibinfo {author} {\bibfnamefont {B.}~\bibnamefont {Marques}}, \ and\ \bibinfo {author} {\bibfnamefont {M.}~\bibnamefont {Bourennane}},\ }\href {\doibase 10.1103/PhysRevLett.114.170502} {\bibfield  {journal} {\bibinfo  {journal} {Phys. Rev. Lett.}\ }\textbf {\bibinfo {volume} {114}},\ \bibinfo {pages} {170502} (\bibinfo {year} {2015})}\BibitemShut {NoStop}%
\bibitem [{\citenamefont {Hirsch}\ \emph {et~al.}(2018)\citenamefont {Hirsch}, \citenamefont {Quintino},\ and\ \citenamefont {Brunner}}]{incombell}%
  \BibitemOpen
  \bibfield  {author} {\bibinfo {author} {\bibfnamefont {F.}~\bibnamefont {Hirsch}}, \bibinfo {author} {\bibfnamefont {M.~T.}\ \bibnamefont {Quintino}}, \ and\ \bibinfo {author} {\bibfnamefont {N.}~\bibnamefont {Brunner}},\ }\href {\doibase 10.1103/PhysRevA.97.012129} {\bibfield  {journal} {\bibinfo  {journal} {Phys. Rev. A}\ }\textbf {\bibinfo {volume} {97}},\ \bibinfo {pages} {012129} (\bibinfo {year} {2018})}\BibitemShut {NoStop}%
\bibitem [{\citenamefont {Loulidi}\ and\ \citenamefont {Nechita}(2022)}]{incombell1}%
  \BibitemOpen
  \bibfield  {author} {\bibinfo {author} {\bibfnamefont {F.}~\bibnamefont {Loulidi}}\ and\ \bibinfo {author} {\bibfnamefont {I.}~\bibnamefont {Nechita}},\ }\href {\doibase 10.1103/PRXQuantum.3.040325} {\bibfield  {journal} {\bibinfo  {journal} {PRX Quantum}\ }\textbf {\bibinfo {volume} {3}},\ \bibinfo {pages} {040325} (\bibinfo {year} {2022})}\BibitemShut {NoStop}%
\bibitem [{\citenamefont {Bene}\ and\ \citenamefont {Vértesi}(2018)}]{Bene2018mi}%
  \BibitemOpen
  \bibfield  {author} {\bibinfo {author} {\bibfnamefont {E.}~\bibnamefont {Bene}}\ and\ \bibinfo {author} {\bibfnamefont {T.}~\bibnamefont {Vértesi}},\ }\href {\doibase 10.1088/1367-2630/aa9ca3} {\bibfield  {journal} {\bibinfo  {journal} {New Journal of Physics}\ }\textbf {\bibinfo {volume} {20}},\ \bibinfo {pages} {013021} (\bibinfo {year} {2018})}\BibitemShut {NoStop}%
\bibitem [{\citenamefont {Quintino}\ \emph {et~al.}(2014)\citenamefont {Quintino}, \citenamefont {V\'ertesi},\ and\ \citenamefont {Brunner}}]{steering2}%
  \BibitemOpen
  \bibfield  {author} {\bibinfo {author} {\bibfnamefont {M.~T.}\ \bibnamefont {Quintino}}, \bibinfo {author} {\bibfnamefont {T.}~\bibnamefont {V\'ertesi}}, \ and\ \bibinfo {author} {\bibfnamefont {N.}~\bibnamefont {Brunner}},\ }\href {\doibase 10.1103/PhysRevLett.113.160402} {\bibfield  {journal} {\bibinfo  {journal} {Phys. Rev. Lett.}\ }\textbf {\bibinfo {volume} {113}},\ \bibinfo {pages} {160402} (\bibinfo {year} {2014})}\BibitemShut {NoStop}%
\bibitem [{\citenamefont {Uola}\ \emph {et~al.}(2015)\citenamefont {Uola}, \citenamefont {Budroni}, \citenamefont {G\"uhne},\ and\ \citenamefont {Pellonp\"a\"a}}]{steering3}%
  \BibitemOpen
  \bibfield  {author} {\bibinfo {author} {\bibfnamefont {R.}~\bibnamefont {Uola}}, \bibinfo {author} {\bibfnamefont {C.}~\bibnamefont {Budroni}}, \bibinfo {author} {\bibfnamefont {O.}~\bibnamefont {G\"uhne}}, \ and\ \bibinfo {author} {\bibfnamefont {J.-P.}\ \bibnamefont {Pellonp\"a\"a}},\ }\href {\doibase 10.1103/PhysRevLett.115.230402} {\bibfield  {journal} {\bibinfo  {journal} {Phys. Rev. Lett.}\ }\textbf {\bibinfo {volume} {115}},\ \bibinfo {pages} {230402} (\bibinfo {year} {2015})}\BibitemShut {NoStop}%
\bibitem [{\citenamefont {Sarkar}\ \emph {et~al.}(2022)\citenamefont {Sarkar}, \citenamefont {Saha},\ and\ \citenamefont {Augusiak}}]{steering1}%
  \BibitemOpen
  \bibfield  {author} {\bibinfo {author} {\bibfnamefont {S.}~\bibnamefont {Sarkar}}, \bibinfo {author} {\bibfnamefont {D.}~\bibnamefont {Saha}}, \ and\ \bibinfo {author} {\bibfnamefont {R.}~\bibnamefont {Augusiak}},\ }\href {\doibase 10.1103/PhysRevA.106.L040402} {\bibfield  {journal} {\bibinfo  {journal} {Phys. Rev. A}\ }\textbf {\bibinfo {volume} {106}},\ \bibinfo {pages} {L040402} (\bibinfo {year} {2022})}\BibitemShut {NoStop}%
\bibitem [{\citenamefont {Buscemi}\ \emph {et~al.}(2020)\citenamefont {Buscemi}, \citenamefont {Chitambar},\ and\ \citenamefont {Zhou}}]{buscemirt1}%
  \BibitemOpen
  \bibfield  {author} {\bibinfo {author} {\bibfnamefont {F.}~\bibnamefont {Buscemi}}, \bibinfo {author} {\bibfnamefont {E.}~\bibnamefont {Chitambar}}, \ and\ \bibinfo {author} {\bibfnamefont {W.}~\bibnamefont {Zhou}},\ }\href {\doibase 10.1103/PhysRevLett.124.120401} {\bibfield  {journal} {\bibinfo  {journal} {Phys. Rev. Lett.}\ }\textbf {\bibinfo {volume} {124}},\ \bibinfo {pages} {120401} (\bibinfo {year} {2020})}\BibitemShut {NoStop}%
\bibitem [{\citenamefont {Clauser}\ \emph {et~al.}(1969)\citenamefont {Clauser}, \citenamefont {Horne}, \citenamefont {Shimony},\ and\ \citenamefont {Holt}}]{chsh}%
  \BibitemOpen
  \bibfield  {author} {\bibinfo {author} {\bibfnamefont {J.~F.}\ \bibnamefont {Clauser}}, \bibinfo {author} {\bibfnamefont {M.~A.}\ \bibnamefont {Horne}}, \bibinfo {author} {\bibfnamefont {A.}~\bibnamefont {Shimony}}, \ and\ \bibinfo {author} {\bibfnamefont {R.~A.}\ \bibnamefont {Holt}},\ }\href {\doibase 10.1103/PhysRevLett.23.880} {\bibfield  {journal} {\bibinfo  {journal} {Phys. Rev. Lett.}\ }\textbf {\bibinfo {volume} {23}},\ \bibinfo {pages} {880} (\bibinfo {year} {1969})}\BibitemShut {NoStop}%
\bibitem [{\citenamefont {Cirel’son}(1980)}]{Cirelson1980}%
  \BibitemOpen
  \bibfield  {author} {\bibinfo {author} {\bibfnamefont {B.~S.}\ \bibnamefont {Cirel’son}},\ }\href {\doibase 10.1007/bf00417500} {\bibfield  {journal} {\bibinfo  {journal} {Letters in Mathematical Physics}\ }\textbf {\bibinfo {volume} {4}},\ \bibinfo {pages} {93–100} (\bibinfo {year} {1980})}\BibitemShut {NoStop}%
\bibitem [{\citenamefont {Wolf}\ \emph {et~al.}(2009)\citenamefont {Wolf}, \citenamefont {Perez-Garcia},\ and\ \citenamefont {Fernandez}}]{wolf2009}%
  \BibitemOpen
  \bibfield  {author} {\bibinfo {author} {\bibfnamefont {M.~M.}\ \bibnamefont {Wolf}}, \bibinfo {author} {\bibfnamefont {D.}~\bibnamefont {Perez-Garcia}}, \ and\ \bibinfo {author} {\bibfnamefont {C.}~\bibnamefont {Fernandez}},\ }\href {\doibase 10.1103/PhysRevLett.103.230402} {\bibfield  {journal} {\bibinfo  {journal} {Phys. Rev. Lett.}\ }\textbf {\bibinfo {volume} {103}},\ \bibinfo {pages} {230402} (\bibinfo {year} {2009})}\BibitemShut {NoStop}%
\bibitem [{\citenamefont {Ghosh}\ \emph {et~al.}(2023)\citenamefont {Ghosh}, \citenamefont {Srivastava}, \citenamefont {Choudhary}, \citenamefont {Lobo},\ and\ \citenamefont {Sen}}]{priya_ghosh2023}%
  \BibitemOpen
  \bibfield  {author} {\bibinfo {author} {\bibfnamefont {P.}~\bibnamefont {Ghosh}}, \bibinfo {author} {\bibfnamefont {C.}~\bibnamefont {Srivastava}}, \bibinfo {author} {\bibfnamefont {S.}~\bibnamefont {Choudhary}}, \bibinfo {author} {\bibfnamefont {E.~P.}\ \bibnamefont {Lobo}}, \ and\ \bibinfo {author} {\bibfnamefont {U.}~\bibnamefont {Sen}},\ }\href {https://arxiv.org/abs/2312.15705} {\enquote {\bibinfo {title} {Measurement incompatibility at remote entangled parties is insufficient for bell nonlocality in two-input and two-output setting},}\ } (\bibinfo {year} {2023}),\ \Eprint {http://arxiv.org/abs/2312.15705} {arXiv:2312.15705 [quant-ph]} \BibitemShut {NoStop}%
\bibitem [{\citenamefont {Collins}\ \emph {et~al.}(2002)\citenamefont {Collins}, \citenamefont {Gisin}, \citenamefont {Linden}, \citenamefont {Massar},\ and\ \citenamefont {Popescu}}]{Collins2002bell}%
  \BibitemOpen
  \bibfield  {author} {\bibinfo {author} {\bibfnamefont {D.}~\bibnamefont {Collins}}, \bibinfo {author} {\bibfnamefont {N.}~\bibnamefont {Gisin}}, \bibinfo {author} {\bibfnamefont {N.}~\bibnamefont {Linden}}, \bibinfo {author} {\bibfnamefont {S.}~\bibnamefont {Massar}}, \ and\ \bibinfo {author} {\bibfnamefont {S.}~\bibnamefont {Popescu}},\ }\href {\doibase 10.1103/PhysRevLett.88.040404} {\bibfield  {journal} {\bibinfo  {journal} {Phys. Rev. Lett.}\ }\textbf {\bibinfo {volume} {88}},\ \bibinfo {pages} {040404} (\bibinfo {year} {2002})}\BibitemShut {NoStop}%
\bibitem [{\citenamefont {Lo}\ \emph {et~al.}(2016)\citenamefont {Lo}, \citenamefont {Li}, \citenamefont {Yabushita}, \citenamefont {Chen}, \citenamefont {Luo},\ and\ \citenamefont {Kobayashi}}]{Lo2016experimental}%
  \BibitemOpen
  \bibfield  {author} {\bibinfo {author} {\bibfnamefont {H.-P.}\ \bibnamefont {Lo}}, \bibinfo {author} {\bibfnamefont {C.-M.}\ \bibnamefont {Li}}, \bibinfo {author} {\bibfnamefont {A.}~\bibnamefont {Yabushita}}, \bibinfo {author} {\bibfnamefont {Y.-N.}\ \bibnamefont {Chen}}, \bibinfo {author} {\bibfnamefont {C.-W.}\ \bibnamefont {Luo}}, \ and\ \bibinfo {author} {\bibfnamefont {T.}~\bibnamefont {Kobayashi}},\ }\href {\doibase 10.1038/srep22088} {\bibfield  {journal} {\bibinfo  {journal} {Scientific Reports}\ }\textbf {\bibinfo {volume} {6}} (\bibinfo {year} {2016}),\ 10.1038/srep22088}\BibitemShut {NoStop}%
\bibitem [{\citenamefont {Zhang}\ \emph {et~al.}(2024)\citenamefont {Zhang}, \citenamefont {Qiu},\ and\ \citenamefont {Chen}}]{Zhang2024experiment}%
  \BibitemOpen
  \bibfield  {author} {\bibinfo {author} {\bibfnamefont {D.}~\bibnamefont {Zhang}}, \bibinfo {author} {\bibfnamefont {X.}~\bibnamefont {Qiu}}, \ and\ \bibinfo {author} {\bibfnamefont {L.}~\bibnamefont {Chen}},\ }\href {\doibase 10.1103/PhysRevA.110.012202} {\bibfield  {journal} {\bibinfo  {journal} {Phys. Rev. A}\ }\textbf {\bibinfo {volume} {110}},\ \bibinfo {pages} {012202} (\bibinfo {year} {2024})}\BibitemShut {NoStop}%
\bibitem [{\citenamefont {Masanes}(2002)}]{cglmp2}%
  \BibitemOpen
  \bibfield  {author} {\bibinfo {author} {\bibfnamefont {L.}~\bibnamefont {Masanes}},\ }\href {\doibase 10.48550/ARXIV.QUANT-PH/0210073} {\  (\bibinfo {year} {2002}),\ 10.48550/ARXIV.QUANT-PH/0210073}\BibitemShut {NoStop}%
\bibitem [{\citenamefont {Chen}\ \emph {et~al.}(2006)\citenamefont {Chen}, \citenamefont {Wu}, \citenamefont {Kwek}, \citenamefont {Oh},\ and\ \citenamefont {Ge}}]{cglmp3}%
  \BibitemOpen
  \bibfield  {author} {\bibinfo {author} {\bibfnamefont {J.-L.}\ \bibnamefont {Chen}}, \bibinfo {author} {\bibfnamefont {C.}~\bibnamefont {Wu}}, \bibinfo {author} {\bibfnamefont {L.~C.}\ \bibnamefont {Kwek}}, \bibinfo {author} {\bibfnamefont {C.~H.}\ \bibnamefont {Oh}}, \ and\ \bibinfo {author} {\bibfnamefont {M.-L.}\ \bibnamefont {Ge}},\ }\href {\doibase 10.1103/PhysRevA.74.032106} {\bibfield  {journal} {\bibinfo  {journal} {Phys. Rev. A}\ }\textbf {\bibinfo {volume} {74}},\ \bibinfo {pages} {032106} (\bibinfo {year} {2006})}\BibitemShut {NoStop}%
\bibitem [{\citenamefont {Fonseca}\ and\ \citenamefont {Parisio}(2015)}]{cglmp4}%
  \BibitemOpen
  \bibfield  {author} {\bibinfo {author} {\bibfnamefont {E.~A.}\ \bibnamefont {Fonseca}}\ and\ \bibinfo {author} {\bibfnamefont {F.}~\bibnamefont {Parisio}},\ }\href {\doibase 10.1103/PhysRevA.92.030101} {\bibfield  {journal} {\bibinfo  {journal} {Phys. Rev. A}\ }\textbf {\bibinfo {volume} {92}},\ \bibinfo {pages} {030101} (\bibinfo {year} {2015})}\BibitemShut {NoStop}%
\bibitem [{\citenamefont {Fonseca}\ \emph {et~al.}(2018)\citenamefont {Fonseca}, \citenamefont {de~Rosier}, \citenamefont {V\'ertesi}, \citenamefont {Laskowski},\ and\ \citenamefont {Parisio}}]{cglmp5}%
  \BibitemOpen
  \bibfield  {author} {\bibinfo {author} {\bibfnamefont {A.}~\bibnamefont {Fonseca}}, \bibinfo {author} {\bibfnamefont {A.}~\bibnamefont {de~Rosier}}, \bibinfo {author} {\bibfnamefont {T.}~\bibnamefont {V\'ertesi}}, \bibinfo {author} {\bibfnamefont {W.}~\bibnamefont {Laskowski}}, \ and\ \bibinfo {author} {\bibfnamefont {F.}~\bibnamefont {Parisio}},\ }\href {\doibase 10.1103/PhysRevA.98.042105} {\bibfield  {journal} {\bibinfo  {journal} {Phys. Rev. A}\ }\textbf {\bibinfo {volume} {98}},\ \bibinfo {pages} {042105} (\bibinfo {year} {2018})}\BibitemShut {NoStop}%
\bibitem [{\citenamefont {Roy}\ \emph {et~al.}(2024)\citenamefont {Roy}, \citenamefont {Kumari}, \citenamefont {Mal},\ and\ \citenamefont {Sen(De)}}]{roy2024robustness}%
  \BibitemOpen
  \bibfield  {author} {\bibinfo {author} {\bibfnamefont {S.}~\bibnamefont {Roy}}, \bibinfo {author} {\bibfnamefont {A.}~\bibnamefont {Kumari}}, \bibinfo {author} {\bibfnamefont {S.}~\bibnamefont {Mal}}, \ and\ \bibinfo {author} {\bibfnamefont {A.}~\bibnamefont {Sen(De)}},\ }\href {\doibase 10.1103/PhysRevA.109.062227} {\bibfield  {journal} {\bibinfo  {journal} {Phys. Rev. A}\ }\textbf {\bibinfo {volume} {109}},\ \bibinfo {pages} {062227} (\bibinfo {year} {2024})}\BibitemShut {NoStop}%
\bibitem [{\citenamefont {Marciniak}\ \emph {et~al.}(2015)\citenamefont {Marciniak}, \citenamefont {Rutkowski}, \citenamefont {Yin}, \citenamefont {Horodecki},\ and\ \citenamefont {Horodecki}}]{marciniak2015}%
  \BibitemOpen
  \bibfield  {author} {\bibinfo {author} {\bibfnamefont {M.}~\bibnamefont {Marciniak}}, \bibinfo {author} {\bibfnamefont {A.}~\bibnamefont {Rutkowski}}, \bibinfo {author} {\bibfnamefont {Z.}~\bibnamefont {Yin}}, \bibinfo {author} {\bibfnamefont {M.}~\bibnamefont {Horodecki}}, \ and\ \bibinfo {author} {\bibfnamefont {R.}~\bibnamefont {Horodecki}},\ }\href {\doibase 10.1103/PhysRevLett.115.170401} {\bibfield  {journal} {\bibinfo  {journal} {Phys. Rev. Lett.}\ }\textbf {\bibinfo {volume} {115}},\ \bibinfo {pages} {170401} (\bibinfo {year} {2015})}\BibitemShut {NoStop}%
\bibitem [{\citenamefont {Srivastav}\ \emph {et~al.}(2022)\citenamefont {Srivastav}, \citenamefont {Valencia}, \citenamefont {McCutcheon}, \citenamefont {Leedumrongwatthanakun}, \citenamefont {Designolle}, \citenamefont {Uola}, \citenamefont {Brunner},\ and\ \citenamefont {Malik}}]{Srivastav2022}%
  \BibitemOpen
  \bibfield  {author} {\bibinfo {author} {\bibfnamefont {V.}~\bibnamefont {Srivastav}}, \bibinfo {author} {\bibfnamefont {N.~H.}\ \bibnamefont {Valencia}}, \bibinfo {author} {\bibfnamefont {W.}~\bibnamefont {McCutcheon}}, \bibinfo {author} {\bibfnamefont {S.}~\bibnamefont {Leedumrongwatthanakun}}, \bibinfo {author} {\bibfnamefont {S.}~\bibnamefont {Designolle}}, \bibinfo {author} {\bibfnamefont {R.}~\bibnamefont {Uola}}, \bibinfo {author} {\bibfnamefont {N.}~\bibnamefont {Brunner}}, \ and\ \bibinfo {author} {\bibfnamefont {M.}~\bibnamefont {Malik}},\ }\href {\doibase 10.1103/PhysRevX.12.041023} {\bibfield  {journal} {\bibinfo  {journal} {Phys. Rev. X}\ }\textbf {\bibinfo {volume} {12}},\ \bibinfo {pages} {041023} (\bibinfo {year} {2022})}\BibitemShut {NoStop}%
\bibitem [{\citenamefont {Neeley}\ \emph {et~al.}(2009)\citenamefont {Neeley}, \citenamefont {Ansmann}, \citenamefont {Bialczak}, \citenamefont {Hofheinz}, \citenamefont {Lucero}, \citenamefont {O'Connell}, \citenamefont {Sank}, \citenamefont {Wang}, \citenamefont {Wenner}, \citenamefont {Cleland}, \citenamefont {Geller},\ and\ \citenamefont {Martinis}}]{Neeley2009}%
  \BibitemOpen
  \bibfield  {author} {\bibinfo {author} {\bibfnamefont {M.}~\bibnamefont {Neeley}}, \bibinfo {author} {\bibfnamefont {M.}~\bibnamefont {Ansmann}}, \bibinfo {author} {\bibfnamefont {R.~C.}\ \bibnamefont {Bialczak}}, \bibinfo {author} {\bibfnamefont {M.}~\bibnamefont {Hofheinz}}, \bibinfo {author} {\bibfnamefont {E.}~\bibnamefont {Lucero}}, \bibinfo {author} {\bibfnamefont {A.~D.}\ \bibnamefont {O'Connell}}, \bibinfo {author} {\bibfnamefont {D.}~\bibnamefont {Sank}}, \bibinfo {author} {\bibfnamefont {H.}~\bibnamefont {Wang}}, \bibinfo {author} {\bibfnamefont {J.}~\bibnamefont {Wenner}}, \bibinfo {author} {\bibfnamefont {A.~N.}\ \bibnamefont {Cleland}}, \bibinfo {author} {\bibfnamefont {M.~R.}\ \bibnamefont {Geller}}, \ and\ \bibinfo {author} {\bibfnamefont {J.~M.}\ \bibnamefont {Martinis}},\ }\href {\doibase 10.1126/science.1173440} {\bibfield  {journal} {\bibinfo  {journal} {Science}\ }\textbf {\bibinfo {volume} {325}},\ \bibinfo {pages} {722} (\bibinfo {year} {2009})}\BibitemShut {NoStop}%
\bibitem [{\citenamefont {Kaltenbaek}\ \emph {et~al.}(2010)\citenamefont {Kaltenbaek}, \citenamefont {Lavoie}, \citenamefont {Zeng}, \citenamefont {Bartlett},\ and\ \citenamefont {Resch}}]{Kaltenbaek2010}%
  \BibitemOpen
  \bibfield  {author} {\bibinfo {author} {\bibfnamefont {R.}~\bibnamefont {Kaltenbaek}}, \bibinfo {author} {\bibfnamefont {J.}~\bibnamefont {Lavoie}}, \bibinfo {author} {\bibfnamefont {B.}~\bibnamefont {Zeng}}, \bibinfo {author} {\bibfnamefont {S.~D.}\ \bibnamefont {Bartlett}}, \ and\ \bibinfo {author} {\bibfnamefont {K.~J.}\ \bibnamefont {Resch}},\ }\href {\doibase 10.1038/nphys1777} {\bibfield  {journal} {\bibinfo  {journal} {Nature Physics}\ }\textbf {\bibinfo {volume} {6}},\ \bibinfo {pages} {850} (\bibinfo {year} {2010})}\BibitemShut {NoStop}%
\bibitem [{\citenamefont {Bechmann-Pasquinucci}\ and\ \citenamefont {Tittel}(2000)}]{bechmann2000}%
  \BibitemOpen
  \bibfield  {author} {\bibinfo {author} {\bibfnamefont {H.}~\bibnamefont {Bechmann-Pasquinucci}}\ and\ \bibinfo {author} {\bibfnamefont {W.}~\bibnamefont {Tittel}},\ }\href {\doibase 10.1103/PhysRevA.61.062308} {\bibfield  {journal} {\bibinfo  {journal} {Phys. Rev. A}\ }\textbf {\bibinfo {volume} {61}},\ \bibinfo {pages} {062308} (\bibinfo {year} {2000})}\BibitemShut {NoStop}%
\bibitem [{\citenamefont {Cerf}\ \emph {et~al.}(2002)\citenamefont {Cerf}, \citenamefont {Bourennane}, \citenamefont {Karlsson},\ and\ \citenamefont {Gisin}}]{cerf_qkd_prl_2002}%
  \BibitemOpen
  \bibfield  {author} {\bibinfo {author} {\bibfnamefont {N.~J.}\ \bibnamefont {Cerf}}, \bibinfo {author} {\bibfnamefont {M.}~\bibnamefont {Bourennane}}, \bibinfo {author} {\bibfnamefont {A.}~\bibnamefont {Karlsson}}, \ and\ \bibinfo {author} {\bibfnamefont {N.}~\bibnamefont {Gisin}},\ }\href {\doibase 10.1103/PhysRevLett.88.127902} {\bibfield  {journal} {\bibinfo  {journal} {Phys. Rev. Lett.}\ }\textbf {\bibinfo {volume} {88}},\ \bibinfo {pages} {127902} (\bibinfo {year} {2002})}\BibitemShut {NoStop}%
\bibitem [{\citenamefont {Nikolopoulos}\ \emph {et~al.}(2006)\citenamefont {Nikolopoulos}, \citenamefont {Ranade},\ and\ \citenamefont {Alber}}]{nikolopoulos2006}%
  \BibitemOpen
  \bibfield  {author} {\bibinfo {author} {\bibfnamefont {G.~M.}\ \bibnamefont {Nikolopoulos}}, \bibinfo {author} {\bibfnamefont {K.~S.}\ \bibnamefont {Ranade}}, \ and\ \bibinfo {author} {\bibfnamefont {G.}~\bibnamefont {Alber}},\ }\href {\doibase 10.1103/PhysRevA.73.032325} {\bibfield  {journal} {\bibinfo  {journal} {Phys. Rev. A}\ }\textbf {\bibinfo {volume} {73}},\ \bibinfo {pages} {032325} (\bibinfo {year} {2006})}\BibitemShut {NoStop}%
\bibitem [{\citenamefont {Sheridan}\ and\ \citenamefont {Scarani}(2010)}]{sheridan2010}%
  \BibitemOpen
  \bibfield  {author} {\bibinfo {author} {\bibfnamefont {L.}~\bibnamefont {Sheridan}}\ and\ \bibinfo {author} {\bibfnamefont {V.}~\bibnamefont {Scarani}},\ }\href {\doibase 10.1103/PhysRevA.82.030301} {\bibfield  {journal} {\bibinfo  {journal} {Phys. Rev. A}\ }\textbf {\bibinfo {volume} {82}},\ \bibinfo {pages} {030301} (\bibinfo {year} {2010})}\BibitemShut {NoStop}%
\bibitem [{\citenamefont {Sasaki}\ \emph {et~al.}(2014)\citenamefont {Sasaki}, \citenamefont {Yamamoto},\ and\ \citenamefont {Koashi}}]{Sasaki2014}%
  \BibitemOpen
  \bibfield  {author} {\bibinfo {author} {\bibfnamefont {T.}~\bibnamefont {Sasaki}}, \bibinfo {author} {\bibfnamefont {Y.}~\bibnamefont {Yamamoto}}, \ and\ \bibinfo {author} {\bibfnamefont {M.}~\bibnamefont {Koashi}},\ }\href {\doibase 10.1038/nature13303} {\bibfield  {journal} {\bibinfo  {journal} {Nature}\ }\textbf {\bibinfo {volume} {509}},\ \bibinfo {pages} {475} (\bibinfo {year} {2014})}\BibitemShut {NoStop}%
\bibitem [{\citenamefont {Bouchard}\ \emph {et~al.}(2018)\citenamefont {Bouchard}, \citenamefont {Heshami}, \citenamefont {England}, \citenamefont {Fickler}, \citenamefont {Boyd}, \citenamefont {Englert}, \citenamefont {Sánchez-Soto},\ and\ \citenamefont {Karimi}}]{Bouchard2018}%
  \BibitemOpen
  \bibfield  {author} {\bibinfo {author} {\bibfnamefont {F.}~\bibnamefont {Bouchard}}, \bibinfo {author} {\bibfnamefont {K.}~\bibnamefont {Heshami}}, \bibinfo {author} {\bibfnamefont {D.}~\bibnamefont {England}}, \bibinfo {author} {\bibfnamefont {R.}~\bibnamefont {Fickler}}, \bibinfo {author} {\bibfnamefont {R.~W.}\ \bibnamefont {Boyd}}, \bibinfo {author} {\bibfnamefont {B.-G.}\ \bibnamefont {Englert}}, \bibinfo {author} {\bibfnamefont {L.~L.}\ \bibnamefont {Sánchez-Soto}}, \ and\ \bibinfo {author} {\bibfnamefont {E.}~\bibnamefont {Karimi}},\ }\href {\doibase 10.22331/q-2018-12-04-111} {\bibfield  {journal} {\bibinfo  {journal} {Quantum}\ }\textbf {\bibinfo {volume} {2}},\ \bibinfo {pages} {111} (\bibinfo {year} {2018})}\BibitemShut {NoStop}%
\bibitem [{\citenamefont {Araújo}\ \emph {et~al.}(2023)\citenamefont {Araújo}, \citenamefont {Huber}, \citenamefont {Navascués}, \citenamefont {Pivoluska},\ and\ \citenamefont {Tavakoli}}]{araujo2023}%
  \BibitemOpen
  \bibfield  {author} {\bibinfo {author} {\bibfnamefont {M.}~\bibnamefont {Araújo}}, \bibinfo {author} {\bibfnamefont {M.}~\bibnamefont {Huber}}, \bibinfo {author} {\bibfnamefont {M.}~\bibnamefont {Navascués}}, \bibinfo {author} {\bibfnamefont {M.}~\bibnamefont {Pivoluska}}, \ and\ \bibinfo {author} {\bibfnamefont {A.}~\bibnamefont {Tavakoli}},\ }\href {\doibase 10.22331/q-2023-05-24-1019} {\bibfield  {journal} {\bibinfo  {journal} {Quantum}\ }\textbf {\bibinfo {volume} {7}},\ \bibinfo {pages} {1019} (\bibinfo {year} {2023})}\BibitemShut {NoStop}%
\bibitem [{\citenamefont {Nagali}\ \emph {et~al.}(2010)\citenamefont {Nagali}, \citenamefont {Giovannini}, \citenamefont {Marrucci}, \citenamefont {Slussarenko}, \citenamefont {Santamato},\ and\ \citenamefont {Sciarrino}}]{nagali2010}%
  \BibitemOpen
  \bibfield  {author} {\bibinfo {author} {\bibfnamefont {E.}~\bibnamefont {Nagali}}, \bibinfo {author} {\bibfnamefont {D.}~\bibnamefont {Giovannini}}, \bibinfo {author} {\bibfnamefont {L.}~\bibnamefont {Marrucci}}, \bibinfo {author} {\bibfnamefont {S.}~\bibnamefont {Slussarenko}}, \bibinfo {author} {\bibfnamefont {E.}~\bibnamefont {Santamato}}, \ and\ \bibinfo {author} {\bibfnamefont {F.}~\bibnamefont {Sciarrino}},\ }\href {\doibase 10.1103/PhysRevLett.105.073602} {\bibfield  {journal} {\bibinfo  {journal} {Phys. Rev. Lett.}\ }\textbf {\bibinfo {volume} {105}},\ \bibinfo {pages} {073602} (\bibinfo {year} {2010})}\BibitemShut {NoStop}%
\bibitem [{\citenamefont {Bouchard}\ \emph {et~al.}(2017)\citenamefont {Bouchard}, \citenamefont {Fickler}, \citenamefont {Boyd},\ and\ \citenamefont {Karimi}}]{Bouchard2017}%
  \BibitemOpen
  \bibfield  {author} {\bibinfo {author} {\bibfnamefont {F.}~\bibnamefont {Bouchard}}, \bibinfo {author} {\bibfnamefont {R.}~\bibnamefont {Fickler}}, \bibinfo {author} {\bibfnamefont {R.~W.}\ \bibnamefont {Boyd}}, \ and\ \bibinfo {author} {\bibfnamefont {E.}~\bibnamefont {Karimi}},\ }\href {\doibase 10.1126/sciadv.1601915} {\bibfield  {journal} {\bibinfo  {journal} {Science Advances}\ }\textbf {\bibinfo {volume} {3}} (\bibinfo {year} {2017}),\ 10.1126/sciadv.1601915}\BibitemShut {NoStop}%
\bibitem [{\citenamefont {Correa}(2014)}]{correra2014}%
  \BibitemOpen
  \bibfield  {author} {\bibinfo {author} {\bibfnamefont {L.~A.}\ \bibnamefont {Correa}},\ }\href {\doibase 10.1103/PhysRevE.89.042128} {\bibfield  {journal} {\bibinfo  {journal} {Phys. Rev. E}\ }\textbf {\bibinfo {volume} {89}},\ \bibinfo {pages} {042128} (\bibinfo {year} {2014})}\BibitemShut {NoStop}%
\bibitem [{\citenamefont {Santos}\ \emph {et~al.}(2019)\citenamefont {Santos}, \citenamefont {Cakmak}, \citenamefont {Campbell},\ and\ \citenamefont {Zinner}}]{santos2019}%
  \BibitemOpen
  \bibfield  {author} {\bibinfo {author} {\bibfnamefont {A.~C.}\ \bibnamefont {Santos}}, \bibinfo {author} {\bibfnamefont {B.}~\bibnamefont {Cakmak}}, \bibinfo {author} {\bibfnamefont {S.}~\bibnamefont {Campbell}}, \ and\ \bibinfo {author} {\bibfnamefont {N.~T.}\ \bibnamefont {Zinner}},\ }\href {\doibase 10.1103/PhysRevE.100.032107} {\bibfield  {journal} {\bibinfo  {journal} {Physical Review E}\ }\textbf {\bibinfo {volume} {100}},\ \bibinfo {pages} {032107} (\bibinfo {year} {2019})}\BibitemShut {NoStop}%
\bibitem [{\citenamefont {Dou}\ \emph {et~al.}(2020)\citenamefont {Dou}, \citenamefont {Wang},\ and\ \citenamefont {Sun}}]{Dou2020}%
  \BibitemOpen
  \bibfield  {author} {\bibinfo {author} {\bibfnamefont {F.-Q.}\ \bibnamefont {Dou}}, \bibinfo {author} {\bibfnamefont {Y.-J.}\ \bibnamefont {Wang}}, \ and\ \bibinfo {author} {\bibfnamefont {J.-A.}\ \bibnamefont {Sun}},\ }\href {\doibase 10.1209/0295-5075/131/43001} {\bibfield  {journal} {\bibinfo  {journal} {EPL (Europhysics Letters)}\ }\textbf {\bibinfo {volume} {131}},\ \bibinfo {pages} {43001} (\bibinfo {year} {2020})}\BibitemShut {NoStop}%
\bibitem [{\citenamefont {Ghosh}\ and\ \citenamefont {Sen(De)}(2022)}]{ghosh2022}%
  \BibitemOpen
  \bibfield  {author} {\bibinfo {author} {\bibfnamefont {S.}~\bibnamefont {Ghosh}}\ and\ \bibinfo {author} {\bibfnamefont {A.}~\bibnamefont {Sen(De)}},\ }\href {\doibase 10.1103/PhysRevA.105.022628} {\bibfield  {journal} {\bibinfo  {journal} {Phys. Rev. A}\ }\textbf {\bibinfo {volume} {105}},\ \bibinfo {pages} {022628} (\bibinfo {year} {2022})}\BibitemShut {NoStop}%
\bibitem [{\citenamefont {Konar}\ \emph {et~al.}(2023)\citenamefont {Konar}, \citenamefont {Ghosh}, \citenamefont {Pal},\ and\ \citenamefont {Sen(De)}}]{konar2023}%
  \BibitemOpen
  \bibfield  {author} {\bibinfo {author} {\bibfnamefont {T.~K.}\ \bibnamefont {Konar}}, \bibinfo {author} {\bibfnamefont {S.}~\bibnamefont {Ghosh}}, \bibinfo {author} {\bibfnamefont {A.~K.}\ \bibnamefont {Pal}}, \ and\ \bibinfo {author} {\bibfnamefont {A.}~\bibnamefont {Sen(De)}},\ }\href {\doibase 10.1103/PhysRevA.107.032602} {\bibfield  {journal} {\bibinfo  {journal} {Phys. Rev. A}\ }\textbf {\bibinfo {volume} {107}},\ \bibinfo {pages} {032602} (\bibinfo {year} {2023})}\BibitemShut {NoStop}%
\bibitem [{\citenamefont {Mordasewicz}\ and\ \citenamefont {Kaniewski}(2022)}]{Mordasewicz2022}%
  \BibitemOpen
  \bibfield  {author} {\bibinfo {author} {\bibfnamefont {K.}~\bibnamefont {Mordasewicz}}\ and\ \bibinfo {author} {\bibfnamefont {J.}~\bibnamefont {Kaniewski}},\ }\href {\doibase 10.1088/1751-8121/ac71eb} {\bibfield  {journal} {\bibinfo  {journal} {Journal of Physics A: Mathematical and Theoretical}\ }\textbf {\bibinfo {volume} {55}},\ \bibinfo {pages} {265302} (\bibinfo {year} {2022})}\BibitemShut {NoStop}%
\bibitem [{\citenamefont {Wiesner}(1983)}]{Wiesner1983}%
  \BibitemOpen
  \bibfield  {author} {\bibinfo {author} {\bibfnamefont {S.}~\bibnamefont {Wiesner}},\ }\href {\doibase 10.1145/1008908.1008920} {\bibfield  {journal} {\bibinfo  {journal} {ACM SIGACT News}\ }\textbf {\bibinfo {volume} {15}},\ \bibinfo {pages} {78} (\bibinfo {year} {1983})}\BibitemShut {NoStop}%
\bibitem [{\citenamefont {Poderini}\ \emph {et~al.}(2020)\citenamefont {Poderini}, \citenamefont {Brito}, \citenamefont {Nery}, \citenamefont {Sciarrino},\ and\ \citenamefont {Chaves}}]{Poderini2020}%
  \BibitemOpen
  \bibfield  {author} {\bibinfo {author} {\bibfnamefont {D.}~\bibnamefont {Poderini}}, \bibinfo {author} {\bibfnamefont {S.}~\bibnamefont {Brito}}, \bibinfo {author} {\bibfnamefont {R.}~\bibnamefont {Nery}}, \bibinfo {author} {\bibfnamefont {F.}~\bibnamefont {Sciarrino}}, \ and\ \bibinfo {author} {\bibfnamefont {R.}~\bibnamefont {Chaves}},\ }\href {\doibase 10.1103/PhysRevResearch.2.043106} {\bibfield  {journal} {\bibinfo  {journal} {Phys. Rev. Res.}\ }\textbf {\bibinfo {volume} {2}},\ \bibinfo {pages} {043106} (\bibinfo {year} {2020})}\BibitemShut {NoStop}%
\bibitem [{\citenamefont {Egelhaaf}\ \emph {et~al.}(2024)\citenamefont {Egelhaaf}, \citenamefont {Pauwels}, \citenamefont {Quintino},\ and\ \citenamefont {Uola}}]{Sophie2024}%
  \BibitemOpen
  \bibfield  {author} {\bibinfo {author} {\bibfnamefont {S.}~\bibnamefont {Egelhaaf}}, \bibinfo {author} {\bibfnamefont {J.}~\bibnamefont {Pauwels}}, \bibinfo {author} {\bibfnamefont {M.~T.}\ \bibnamefont {Quintino}}, \ and\ \bibinfo {author} {\bibfnamefont {R.}~\bibnamefont {Uola}},\ }\href {https://arxiv.org/abs/2407.06787} {\enquote {\bibinfo {title} {Certifying measurement incompatibility in prepare-and-measure and bell scenarios},}\ } (\bibinfo {year} {2024}),\ \Eprint {http://arxiv.org/abs/2407.06787} {arXiv:2407.06787 [quant-ph]} \BibitemShut {NoStop}%
\bibitem [{\citenamefont {Henaut}\ \emph {et~al.}(2018)\citenamefont {Henaut}, \citenamefont {Catani}, \citenamefont {Browne}, \citenamefont {Mansfield},\ and\ \citenamefont {Pappa}}]{xor1}%
  \BibitemOpen
  \bibfield  {author} {\bibinfo {author} {\bibfnamefont {L.}~\bibnamefont {Henaut}}, \bibinfo {author} {\bibfnamefont {L.}~\bibnamefont {Catani}}, \bibinfo {author} {\bibfnamefont {D.~E.}\ \bibnamefont {Browne}}, \bibinfo {author} {\bibfnamefont {S.}~\bibnamefont {Mansfield}}, \ and\ \bibinfo {author} {\bibfnamefont {A.}~\bibnamefont {Pappa}},\ }\href {\doibase 10.1103/PhysRevA.98.060302} {\bibfield  {journal} {\bibinfo  {journal} {Phys. Rev. A}\ }\textbf {\bibinfo {volume} {98}},\ \bibinfo {pages} {060302} (\bibinfo {year} {2018})}\BibitemShut {NoStop}%
\bibitem [{\citenamefont {Catani}\ \emph {et~al.}(2024)\citenamefont {Catani}, \citenamefont {Faleiro}, \citenamefont {Emeriau}, \citenamefont {Mansfield},\ and\ \citenamefont {Pappa}}]{xor2}%
  \BibitemOpen
  \bibfield  {author} {\bibinfo {author} {\bibfnamefont {L.}~\bibnamefont {Catani}}, \bibinfo {author} {\bibfnamefont {R.}~\bibnamefont {Faleiro}}, \bibinfo {author} {\bibfnamefont {P.-E.}\ \bibnamefont {Emeriau}}, \bibinfo {author} {\bibfnamefont {S.}~\bibnamefont {Mansfield}}, \ and\ \bibinfo {author} {\bibfnamefont {A.}~\bibnamefont {Pappa}},\ }\href {\doibase 10.1103/PhysRevA.109.012427} {\bibfield  {journal} {\bibinfo  {journal} {Phys. Rev. A}\ }\textbf {\bibinfo {volume} {109}},\ \bibinfo {pages} {012427} (\bibinfo {year} {2024})}\BibitemShut {NoStop}%
\bibitem [{\citenamefont {Saha}\ \emph {et~al.}(2023)\citenamefont {Saha}, \citenamefont {Das}, \citenamefont {Das}, \citenamefont {Bhattacharya},\ and\ \citenamefont {Majumdar}}]{archan}%
  \BibitemOpen
  \bibfield  {author} {\bibinfo {author} {\bibfnamefont {D.}~\bibnamefont {Saha}}, \bibinfo {author} {\bibfnamefont {D.}~\bibnamefont {Das}}, \bibinfo {author} {\bibfnamefont {A.~K.}\ \bibnamefont {Das}}, \bibinfo {author} {\bibfnamefont {B.}~\bibnamefont {Bhattacharya}}, \ and\ \bibinfo {author} {\bibfnamefont {A.~S.}\ \bibnamefont {Majumdar}},\ }\href {\doibase 10.1103/PhysRevA.107.062210} {\bibfield  {journal} {\bibinfo  {journal} {Phys. Rev. A}\ }\textbf {\bibinfo {volume} {107}},\ \bibinfo {pages} {062210} (\bibinfo {year} {2023})}\BibitemShut {NoStop}%
\bibitem [{\citenamefont {Runarsson}\ and\ \citenamefont {Yao}(2005)}]{Runarsson2005}%
  \BibitemOpen
  \bibfield  {author} {\bibinfo {author} {\bibfnamefont {T.}~\bibnamefont {Runarsson}}\ and\ \bibinfo {author} {\bibfnamefont {X.}~\bibnamefont {Yao}},\ }\href {\doibase 10.1109/TSMCC.2004.841906} {\bibfield  {journal} {\bibinfo  {journal} {IEEE Transactions on Systems, Man and Cybernetics, Part C (Applications and Reviews)}\ }\textbf {\bibinfo {volume} {35}},\ \bibinfo {pages} {233} (\bibinfo {year} {2005})}\BibitemShut {NoStop}%
\bibitem [{\citenamefont {Ac\'{\i}n}\ \emph {et~al.}(2002)\citenamefont {Ac\'{\i}n}, \citenamefont {Durt}, \citenamefont {Gisin},\ and\ \citenamefont {Latorre}}]{acin2002}%
  \BibitemOpen
  \bibfield  {author} {\bibinfo {author} {\bibfnamefont {A.}~\bibnamefont {Ac\'{\i}n}}, \bibinfo {author} {\bibfnamefont {T.}~\bibnamefont {Durt}}, \bibinfo {author} {\bibfnamefont {N.}~\bibnamefont {Gisin}}, \ and\ \bibinfo {author} {\bibfnamefont {J.~I.}\ \bibnamefont {Latorre}},\ }\href {\doibase 10.1103/PhysRevA.65.052325} {\bibfield  {journal} {\bibinfo  {journal} {Phys. Rev. A}\ }\textbf {\bibinfo {volume} {65}},\ \bibinfo {pages} {052325} (\bibinfo {year} {2002})}\BibitemShut {NoStop}%
\bibitem [{\citenamefont {Zyczkowski}\ and\ \citenamefont {Kus}(1994)}]{Zyczkowski1994}%
  \BibitemOpen
  \bibfield  {author} {\bibinfo {author} {\bibfnamefont {K.}~\bibnamefont {Zyczkowski}}\ and\ \bibinfo {author} {\bibfnamefont {M.}~\bibnamefont {Kus}},\ }\href {\doibase 10.1088/0305-4470/27/12/028} {\bibfield  {journal} {\bibinfo  {journal} {Journal of Physics A: Mathematical and General}\ }\textbf {\bibinfo {volume} {27}},\ \bibinfo {pages} {4235} (\bibinfo {year} {1994})}\BibitemShut {NoStop}%
\bibitem [{\citenamefont {de~Guise}\ \emph {et~al.}(2018)\citenamefont {de~Guise}, \citenamefont {Di~Matteo},\ and\ \citenamefont {S\'anchez-Soto}}]{deGuise2018}%
  \BibitemOpen
  \bibfield  {author} {\bibinfo {author} {\bibfnamefont {H.}~\bibnamefont {de~Guise}}, \bibinfo {author} {\bibfnamefont {O.}~\bibnamefont {Di~Matteo}}, \ and\ \bibinfo {author} {\bibfnamefont {L.~L.}\ \bibnamefont {S\'anchez-Soto}},\ }\href {\doibase 10.1103/PhysRevA.97.022328} {\bibfield  {journal} {\bibinfo  {journal} {Phys. Rev. A}\ }\textbf {\bibinfo {volume} {97}},\ \bibinfo {pages} {022328} (\bibinfo {year} {2018})}\BibitemShut {NoStop}%
\bibitem [{\citenamefont {\ifmmode~\dot{Z}\else \.{Z}\fi{}ukowski}\ \emph {et~al.}(1997)\citenamefont {\ifmmode~\dot{Z}\else \.{Z}\fi{}ukowski}, \citenamefont {Zeilinger},\ and\ \citenamefont {Horne}}]{Zukowski1997realizable}%
  \BibitemOpen
  \bibfield  {author} {\bibinfo {author} {\bibfnamefont {M.}~\bibnamefont {\ifmmode~\dot{Z}\else \.{Z}\fi{}ukowski}}, \bibinfo {author} {\bibfnamefont {A.}~\bibnamefont {Zeilinger}}, \ and\ \bibinfo {author} {\bibfnamefont {M.~A.}\ \bibnamefont {Horne}},\ }\href {\doibase 10.1103/PhysRevA.55.2564} {\bibfield  {journal} {\bibinfo  {journal} {Phys. Rev. A}\ }\textbf {\bibinfo {volume} {55}},\ \bibinfo {pages} {2564} (\bibinfo {year} {1997})}\BibitemShut {NoStop}%
\bibitem [{\citenamefont {Ambainis}\ \emph {et~al.}(1999)\citenamefont {Ambainis}, \citenamefont {Nayak}, \citenamefont {Ta-Shma},\ and\ \citenamefont {Vazirani}}]{Ambainis1999}%
  \BibitemOpen
  \bibfield  {author} {\bibinfo {author} {\bibfnamefont {A.}~\bibnamefont {Ambainis}}, \bibinfo {author} {\bibfnamefont {A.}~\bibnamefont {Nayak}}, \bibinfo {author} {\bibfnamefont {A.}~\bibnamefont {Ta-Shma}}, \ and\ \bibinfo {author} {\bibfnamefont {U.}~\bibnamefont {Vazirani}},\ }in\ \href {\doibase 10.1145/301250.301347} {\emph {\bibinfo {booktitle} {Proceedings of the thirty-first annual ACM symposium on Theory of Computing}}}\ (\bibinfo  {publisher} {ACM},\ \bibinfo {year} {1999})\ pp.\ \bibinfo {pages} {376--383}\BibitemShut {NoStop}%
\bibitem [{\citenamefont {Ambainis}\ \emph {et~al.}(2009)\citenamefont {Ambainis}, \citenamefont {Leung}, \citenamefont {Mancinska},\ and\ \citenamefont {Ozols}}]{ambainis2009quantumrandomaccesscodes}%
  \BibitemOpen
  \bibfield  {author} {\bibinfo {author} {\bibfnamefont {A.}~\bibnamefont {Ambainis}}, \bibinfo {author} {\bibfnamefont {D.}~\bibnamefont {Leung}}, \bibinfo {author} {\bibfnamefont {L.}~\bibnamefont {Mancinska}}, \ and\ \bibinfo {author} {\bibfnamefont {M.}~\bibnamefont {Ozols}},\ }\href {https://arxiv.org/abs/0810.2937} {\enquote {\bibinfo {title} {Quantum random access codes with shared randomness},}\ } (\bibinfo {year} {2009}),\ \Eprint {http://arxiv.org/abs/0810.2937} {arXiv:0810.2937 [quant-ph]} \BibitemShut {NoStop}%
\bibitem [{\citenamefont {Farkas}\ and\ \citenamefont {Kaniewski}(2019)}]{Farkas_2019}%
  \BibitemOpen
  \bibfield  {author} {\bibinfo {author} {\bibfnamefont {M.}~\bibnamefont {Farkas}}\ and\ \bibinfo {author} {\bibfnamefont {J.}~\bibnamefont {Kaniewski}},\ }\href {\doibase 10.1103/physreva.99.032316} {\bibfield  {journal} {\bibinfo  {journal} {Physical Review A}\ }\textbf {\bibinfo {volume} {99}} (\bibinfo {year} {2019}),\ 10.1103/physreva.99.032316}\BibitemShut {NoStop}%
\bibitem [{\citenamefont {Designolle}\ \emph {et~al.}(2019)\citenamefont {Designolle}, \citenamefont {Skrzypczyk}, \citenamefont {Fr\"owis},\ and\ \citenamefont {Brunner}}]{Designolle2019}%
  \BibitemOpen
  \bibfield  {author} {\bibinfo {author} {\bibfnamefont {S.}~\bibnamefont {Designolle}}, \bibinfo {author} {\bibfnamefont {P.}~\bibnamefont {Skrzypczyk}}, \bibinfo {author} {\bibfnamefont {F.}~\bibnamefont {Fr\"owis}}, \ and\ \bibinfo {author} {\bibfnamefont {N.}~\bibnamefont {Brunner}},\ }\href {\doibase 10.1103/PhysRevLett.122.050402} {\bibfield  {journal} {\bibinfo  {journal} {Phys. Rev. Lett.}\ }\textbf {\bibinfo {volume} {122}},\ \bibinfo {pages} {050402} (\bibinfo {year} {2019})}\BibitemShut {NoStop}%
\bibitem [{\citenamefont {Landau}(1987)}]{Landau1987}%
  \BibitemOpen
  \bibfield  {author} {\bibinfo {author} {\bibfnamefont {L.~J.}\ \bibnamefont {Landau}},\ }\href {\doibase 10.1016/0375-9601(87)90075-2} {\bibfield  {journal} {\bibinfo  {journal} {Physics Letters A}\ }\textbf {\bibinfo {volume} {120}},\ \bibinfo {pages} {54} (\bibinfo {year} {1987})}\BibitemShut {NoStop}%
\bibitem [{\citenamefont {Gil}(2018)}]{Gil2018}%
  \BibitemOpen
  \bibfield  {author} {\bibinfo {author} {\bibfnamefont {J.~J.}\ \bibnamefont {Gil}},\ }\href {\doibase 10.1140/epjp/i2018-12032-0} {\bibfield  {journal} {\bibinfo  {journal} {The European Physical Journal Plus}\ }\textbf {\bibinfo {volume} {133}},\ \bibinfo {pages} {206} (\bibinfo {year} {2018})}\BibitemShut {NoStop}%
\end{thebibliography}%

\appendix

\section{Characterizing  incompatibility}
\label{sec:incom}

Measurement incompatibility refers to the existence of quantum measurements that cannot be performed simultaneously on a single system. Consider a collection of Positive Operator-Valued Measures (POVMs), denoted as \(M^a\equiv\{M^a_j\}_j\), where each element satisfies the following conditions: \(M^a_j \geq 0\), meaning each operator is positive semi-definite, and \(\sum_{j} M^a_j = \mathbb{I}\), where \(\mathbb{I}\) is the identity operator.

\subsection{Commutation-based incompatibility quantifier}



In Ref. \cite{Mordasewicz2022}, a new family of incompatibility measures has been defined. The incompatibility of two measurements $M^1\equiv\{M^1_i\}$ and $M^2\equiv\{M^2_j\}$, acting on $\mathbb{C}^d$, is defined as the Schatten $p$-norms of commutator between the elements of the above set of measurements, which can be written mathematically as
\begin{eqnarray}
    \mathcal{I}_p(M^1,M^2)= \sum_{j=1}^{d_1} \sum_{k=1}^{d_2} \|[M_j^1,M_k^2]\|_p,
\label{eq:incom_p_norm}
\end{eqnarray}
where $p\in[1,\infty]$. For all possible values of $p$, these incompatibility measures are non-negative. When all measurement operators commute between each other, only then $\mathcal{I}_p = 0$ $\forall p$. Therefore, it is evident that if $\mathcal{I}_p = 0$ in the case of projective measurements then the two observables corresponding to two measurements commute. The incompatibility quantifier is upper bounded by $\mathcal{I}_p(M^1,M^2)\leq 2^{\frac{1}{p}}d\sqrt{d-1}$. A pair of MUBs saturate the upper bound in $\mathbb{C}^d$ \cite{Mordasewicz2022}.

When the measurements considered are projective, we get the following result:
\begin{proposition}
     Two sets of projectors $\Pi^A = \{\Pi_i^A\}_{i=1}^{N_A}$ and $\Pi^B =\{\Pi_j^B\}_{j=1}^{N_B}$  are jointly measurable if and only if $\mathcal{I}_p(\Pi^A,\Pi^B) = 0$ for all $p$.
\end{proposition}
\begin{proof}
Before proceeding with the proof we first define the {\tt spectral support} for any arbitrary observable.  
\begin{definition}
    The {\tt Spectral Support} of any observable $\mathcal{O}$ is the set $\mathcal{S}_\mathcal{O}$ consisting of all the projectors in the spectral decomposition of $\mathcal{O}$. 
    \label{def:specsupp}
\end{definition}
\noindent Now consider two sets of projectors $\Pi^A = \{\Pi_i^A\}_{i=1}^{N_A}$ and $\Pi^B =\{\Pi_j^B\}_{j=1}^{N_B}$.
Let us take two bounded observables $O^A$ and $O^B$ that satisfy
\begin{eqnarray}
 {\tt Span}   (\mathcal{S}_{O^{A}}) \subseteq {\tt Span}(\Pi^{A}), ~{\tt Span}(\mathcal{S}_{O^{B}}) \subseteq {\tt Span}(\Pi^{B}).
    \label{eq:ssup}
\end{eqnarray}
The most general form of $O^A$ and $O^B$ satisfying Eq. \eqref{eq:ssup} is given by
\begin{eqnarray}
    O^A = \sum_{i=1}^{N_A} \lambda_i^A \Pi_i^A, ~~~
    O^B = \sum_{j=1}^{N_B} \lambda_j^B \Pi_j^B,
    \label{eq:spectral}
\end{eqnarray}
where $\lambda_k^{a(b)}$ are arbitrary complex numbers. 

We know that $O^A$ and $O^B$ are jointly measurable if and only if $[O^A,O^B]= 0$. Using Eq. \eqref{eq:spectral}, the condition of joint measurability between $O^A$ and $O^B$ can be rewritten as 
\begin{eqnarray}
    [O^A,O^B]= \sum_{i,j=1}^{N_A,N_B} \lambda_i^A\lambda_j^B [\Pi_i^A,\Pi_j^B] = 0.
\end{eqnarray}
If $\Pi^A$ and $\Pi^B$ are jointly measurable, then the above equation holds for all $O^A$ and $O^B$ that satisfy Eq. \eqref{eq:ssup}, i.e., for all $\lambda_k^{A(B)}$s. 
Physically, this follows from the fact that all observables that can be measured by jointly measurable instruments (projectors) must be jointly measurable.
This is possible if and only if
\begin{eqnarray}
    [\Pi_i^A,\Pi_j^B] = 0 ~~\forall i,j.
\end{eqnarray}
This right away implies 
\begin{eqnarray}
     [O^A,O^B]= 0 \iff [\Pi_i^A,\Pi_j^B] = 0 ~~~\forall i,j.
     \label{eq:iff}
\end{eqnarray}
Now, from Eq. \eqref{eq:incom_p_norm}, we have $\mathcal{I}_p(\Pi^A,\Pi^B) = 0$ for all $p \in (0,\infty]$. This completes the proof of joint measureability of $\Pi^A$ and $\Pi^B$ $\implies$ $\mathcal{I}_p(\Pi^A,\Pi^B) = 0$. 

Next, we attempt to prove the converse. We consider two arbitrary sets of projective measurements $\Pi^A$ and $\Pi^B$, such that $\mathcal{I}_p(\Pi^A,\Pi^B) = 0$ for all $p \in (0,\infty]$. 
Then for all observables satisfying the condition in Eq. \eqref{eq:ssup}, from Eq. \eqref{eq:iff}, we have $O^A$ and $O^B$ to be jointly measurable. Therefore, $\Pi^A$ and $\Pi^B$ must be jointly measurable. This completes the proof.

\end{proof}
Next, we briefly discuss the notion of joint measureability for generalized measurements. 

\subsection{Robustness of incompatibility}

The POVMs \(M_a\) are said to be jointly measurable or compatible if there exists a parent POVM \(\{G_\lambda\}\) and a set of conditional probabilities \(\{p(j|a, \lambda)\}\) such that the following compatibility condition holds for all outcomes \(j\) and measurement settings \(a\):
\begin{eqnarray}
    M_j^a = \sum_{\lambda} p(j|a, \lambda) G_\lambda,
    \label{eq:compatible_measurement}
\end{eqnarray}
where \(\lambda\) represents some hidden variable. In the absence of such a relationship, the collection of POVMs is considered incompatible. Consider the noisy POVMs \(M^a_{j,\eta}\) and \(N^a_{i,\eta}\) with effects \(M^a_{j,\eta} = \eta M^a_{j} + (1 - \eta) \operatorname{tr}M^a_{j} \, \frac{\mathbb{I}}{d}\) (similarly for \(N^a_{i,\eta}\)). Then, the incompatibility robustness \(\eta^*\) is the largest \(\eta\) such that \(M^a_{j,\eta}\) and \(N^a_{i,\eta}\) are jointly measurable (that is, there exists a parent POVM, \(\{G_{ji}\}_{j,i=1}^d\), where \(d_1 = d_2 = l\) such that \(\sum_i G_{ji} = M^a_{j,\eta}\) and \(\sum_j G_{ji} = N^a_{i,\eta}\)). Recently an analytic upper bound on \(\eta^*\) has been derived for an arbitrary set of POVMs \cite{Designolle2019}. For a pair of POVMs the bound read as
\begin{widetext}
\begin{eqnarray}
    \eta^{*}\leq \frac{l^2 \text{max}_{ji}\|M^a_{j} , N^a_{i}\|-\sum_{j}(\text{tr}M^a_{j})^2 -\sum_{i}(\text{tr}N^a_{i})^2}{l \sum_{j}\text{tr}(M^a_{j})^2 + d\sum_{i}\text{tr}(N^a_{i})^2 -\sum_{j}(\text{tr}M^a_{j})^2 -\sum_{i}(\text{tr}N^a_{i})^2}.
    \label{eq:upper_bound_robustness}
\end{eqnarray}
\end{widetext}

\section{{Proof of Proposition \ref{prop:necessary}}}
\label{app:necessary}
We first notice that given a projective measurement $M^1\equiv\{M^1_i\}_i$ the only compatible measurement is $M^1$ itself according to the definition of Schatten $p$-norm of incompatibility (Eq.~\eqref{eq:incom_p_norm_1}). For bipartite qutrit system, the CGLMP operator, $\mathbf{C}_3$, can be written as
\begin{eqnarray}
    \nonumber \mathbf{C}_3 &=& (A^1_0-A^1_2+A^2_2-A^2_0)B^1_0 + (A^1_1-A^1_0+A^2_0-A^2_1)B^1_1 \\ \nonumber &+& (A^1_2-A^1_1+A^2_1-A^2_2)B^1_2 + (A^2_0-A^2_2+A^1_0-A^1_1)B^2_0 \\ \nonumber &+& (A^2_1-A^2_0+A^1_1-A^1_2)B^2_1 + (A^2_2-A^2_1+A^1_2-A^1_0)B^2_2. \\ 
\end{eqnarray}
If measurements on Alice's side are compatible, both the measurement choices have the same set of projectors.
Let us first choose $A^1_0=A^2_0$, $A^1_1=A^2_1$, $A^1_2=A^2_2$. In this situation, the CGLMP operator reduces to
\begin{eqnarray}
   \mathbf{C}_3 = 3\mathbb{H} - \mathbb{I},
   \label{eq:comp_c3}
\end{eqnarray}
where $\mathbb{H} = A^1_0 B^2_0 + A^1_1 B^2_1 + A^1_2 B^2_2$ and $\mathbb{H}^2 = \mathbb{H}$. Therefore, the maximum possible value of $\mathbf{C}_3$ can be $2$.
Similarly, if Alice performs incompatible projective measurements, but Bob's projective measurements are compatible, i.e., $B^1_0=B^2_0$, $B^1_1=B^2_1$, $B^1_2=B^2_2$, then also $\mathbf{C}_3$ follows Eq.~\eqref{eq:comp_c3}. 

By relabeling, if we choose either $\{A^1_0=A^2_2$, $A^1_1=A^2_0$, $A^1_2=A^2_1\}$ or $\{B^1_0=B^2_2$, $B^1_1=B^2_0$, $B^1_2=B^2_1\}$ then we find $\mathbf{C}_3 = 3\tilde{\mathbb{H}}-\mathbb{I}$ where $\tilde{\mathbb{H}}$ have different expression compared to $\mathbb{H}_3$ but satisfies $\tilde{\mathbb{H}}^2=\tilde{\mathbb{H}}$.

On the other hand, for the case $A^1_0=A^2_1$, $A^1_1=A^2_2$, $A^1_2=A^2_0$, the CGLMP operator follows
\begin{eqnarray}
    \mathbf{C}_3 = 2\mathbb{I} - 3\mathbb{H}_1 - 3\mathbb{H}_2,
\end{eqnarray}
with $\mathbb{H}_1 = A^1_2 B^1_0 + A^1_0 B^1_1 + A^1_1 B^1_2$ and $\mathbb{H}_2 = A^1_2 B^2_1 + A^1_0 B^2_2 + A^1_1 B^2_0$ and $\mathbb{H}_i^2=\mathbb{H}_i$ for $i=1,2$. Notice that, $\mathbb{H}_1$ and $\mathbb{H}_2$ are positive projective operators. Therefore, $\mathbf{C}_3$ can have eigenvalues less than or equal to $2$.
Similarly, by choosing $B^1_0=B^2_1$, $B^1_1=B^2_2$, $B^1_2=B^2_0$ we get $\mathbf{C}_3 = 2\mathbb{I} - 3\mathbb{H}_1 - 3\tilde{\mathbb{H}}_2$ with $\tilde{\mathbb{H}}_2^2=\tilde{\mathbb{H}}_2$.

Therefore, if any one of the parties, Alice or Bob, or both of them perform compatible projective measurements then no violation of the CGLMP inequality can be obtained.

\section{Relationship between CHSH violation and incompatibility}
\label{app:chi2(I)}

Consider two sets of qubit projective measurements 
\begin{eqnarray}
   E &=& \{ \ketbra{e}{e}, \ketbra{\bar{e}}{\bar{e}}\}, \nonumber \\
   F &=& \{ \ketbra{f}{f}, \ketbra{\bar{f}}{\bar{f}}\}.
\end{eqnarray}
Without loss of any generality, we align our axes such that $\ket{e} = \ket{0}$, i.e., along the $z$-axis, and $\ket{f}$ to be in the $x-z$-plane. Therefore, we can parameterize 
\begin{eqnarray}
    \ket{f} = \cos \frac{\theta}{2}\ket{0} + \sin \frac{\theta}{2}\ket{1}.
\end{eqnarray}
Now using Eq. \eqref{eq:ipf(x)} and noting $|\langle e|f\rangle| = \cos \frac{\theta}{2}$, we have
\begin{eqnarray}
    \mathcal{I}_p(E,F) = 2^{\frac1p} 2 \sin \theta.
    \label{eq:ipf(theta)}
\end{eqnarray}
Let us construct two observables $A_0$ and $A_1$ such that
\begin{eqnarray}
    O_E &=& \ketbra{e}{e} - \ketbra{\bar{e}}{\bar{e}} = \sigma_z = \hat{n}_z.\vec{\sigma}, \nonumber \\
    O_F &=& \ketbra{f}{f} - \ketbra{\bar{f}}{\bar{f}}   = \hat{n}_{\theta}.\vec{\sigma},
\end{eqnarray}
where $\hat{n}_z = (0,0,1)^T$ and $\hat{n}_\theta = (\sin \theta,0,\cos \theta)^T$. Now, we have
\begin{eqnarray}
    [O_E,O_F] = [\hat{n}_z.\vec{\sigma},\hat{n}_\theta.\vec{\sigma}] = 2 i\sin \theta \sigma_y.
\end{eqnarray}
The above analysis leads to the following identity
\begin{eqnarray}
    ||[O_E,O_F]|| = 2 \sin \theta = \frac{\mathcal{I}_p(E,F)}{2^\frac1p}.
    \label{eq:comnorm}
\end{eqnarray}

Now, in a bipartite setting, if we consider two observables each for Alice $(A^1,A^2)$ and Bob $(B^1,B^2)$ in $d = 2$, the maximal value of the CHSH inequality can be expressed as \cite{Landau1987}
\begin{eqnarray}
    \chi_2 = \sqrt{4 + ||[A^1,A^2]||\times||[B^1,B^2]||}.
\end{eqnarray}
From Eq. \eqref{eq:comnorm} using Def. \ref{def:specsupp}, we can write the above expression as
\begin{eqnarray}
    \chi_2 = 2\sqrt{1 + \frac{\mathcal{I}_p(\mathcal{S}_{A^1},\mathcal{S}_{A^2})\mathcal{I}_p(\mathcal{S}_{B^1},\mathcal{S}_{B^2})}{4.2^\frac2p}},
\end{eqnarray}
whereas in Def. \ref{def:specsupp}, $\mathcal{S}_O$ is the {\tt spectral support} of $O$, i.e.,   $\mathcal{S}_O$ is the set containing all the projectors in the spectral decomposition of $\mathcal{O}$.

\section{Parameterization of local unitary, U(3)}
\label{sec:d3_local_unitary}

The local unitary \(U(3)\)  can be parameterized as follows \cite{Gil2018},

\begin{widetext}
   \begin{eqnarray}
U_{i}(\phi_i, \theta_i, \varphi_i, \chi_i, \mu_i, \alpha^{i}_1, \alpha^{i}_2, \alpha^{i}_3, \beta^{i}_2)
= \mathbf{Q}_i  
\begin{pmatrix}
 e^{i \alpha^{i}_1}\cos \chi_i & i e^{i \alpha^{i}_2} \cos \mu_i \sin \chi_i & i e^{i \alpha^{i}_3} \sin \mu_i \sin \chi_i \\
i e^{i \alpha^{i}_1} \sin \chi_{i} & e^{i \alpha^{i}_2} \cos \mu_{i} \cos \chi_i & e^{i \alpha^{i}_3} \sin \mu_i \cos \chi_i \\
0 & e^{i \beta^{i}_2} \sin \mu_i & -e^{i (\beta^{i}_2 - \alpha^{i}_2 + \alpha^{i}_3)} \cos \mu_i
\end{pmatrix} 
\mathbf{Q}_{i}^T\nonumber\\
\end{eqnarray} 
where \(\mathbf{Q}_i=\mathbf{Q}_{\phi_i} \mathbf{Q}_{\theta_i} \mathbf{Q}_{\varphi_i}\), and 

\[
\mathbf{Q}_{\varphi_i} \equiv 
\begin{pmatrix}
\cos\varphi_i & -\sin\varphi_i & 0 \\
\sin\varphi_i & \cos\varphi_i & 0 \\
0 & 0 & 1
\end{pmatrix}, \quad
\mathbf{Q}_{\theta_i} \equiv 
\begin{pmatrix}
\cos\theta_i & 0 & -\sin\theta_i \\
0 & 1 & 0 \\
\sin\theta_i & 0 & \cos\theta_i
\end{pmatrix}, \quad
\mathbf{Q}_{\phi_i} \equiv 
\begin{pmatrix}
\cos\phi_i & -\sin\phi_i & 0 \\
\sin\phi_i & \cos\phi_i & 0 \\
0 & 0 & 1
\end{pmatrix}.
\]

and the range of the parameters, \(
-\pi < \phi_i \leq \pi, \quad -\pi/2 \leq \theta_i \leq \pi/2, \quad 0 \leq \varphi_i \leq \pi, \quad -\pi/4 \leq \chi_i \leq \pi/4,\quad 0 \leq \mu_i \leq \pi/2, \quad 0 \leq \alpha^{i}_1, \alpha^{i}_2, \alpha^{i}_3, \beta^{i}_2 \leq \pi\
\) 
\end{widetext}

\section{Average success probability in the $2\to1$ RAC game}
\label{app:rac221}
Let us consider two arbitrary sets of projective measurements for a single qubit as
\begin{eqnarray}
    B^1 &=& \{\ketbra{e}{e},\ketbra{\bar{e}}{\bar{e}}\}, \nonumber \\
    B^2 &=& \{\ketbra{f}{f},\ketbra{\bar{f}}{\bar{f}}\},
\end{eqnarray}
where $\ketbra{\bar{e}}{\bar{e}}=\mathbb{I}-\ketbra{e}{e}$ and $\ketbra{\bar{f}}{\bar{f}}=\mathbb{I}-\ketbra{f}{f}$. 
To proceed, we invoke the following identities \cite{Mordasewicz2022}
\begin{eqnarray}
        ||[\ketbra{q}{q},\ketbra{r}{r}]||_p &=& 2^{\frac1p} c\sqrt{1-c^2}, \nonumber \\
    ||\ketbra{q}{q}+\ketbra{r}{r}||_{\infty} &=& 1+c, 
    \label{eq:formulaoverlap}
\end{eqnarray}
where $c=|\langle q|r\rangle|$.
Using the features 
\begin{eqnarray}
   |\langle e|f\rangle| &=& |\langle \bar{e}|\bar{f}\rangle| = x, \nonumber\\ 
   |\langle e|\bar{f}\rangle| &=& |\langle \bar{e}|f\rangle| = \sqrt{1-x^2},
   \label{eq:formulaoverlap2}
\end{eqnarray}
for the measurements $B_1$ and $B_2$ and Eq. \eqref{eq:formulaoverlap}, the average success probability of the QRAC game, when two arbitrary projectors $B^1$ and $B^2$ are used for decoding, from Eq. \eqref{eq:racsuccessprob} is given by
\begin{eqnarray}
    \mathcal{R}^Q_2 = \frac14 (2+ x + \sqrt{1-x^2}).
    \label{eq:rq}
\end{eqnarray}
On the other hand, noting Eqs. \eqref{eq:formulaoverlap} and \eqref{eq:formulaoverlap2}, the incompatibility between $B^1$ and $B^2$ using Eq. \eqref{eq:incom_p_norm_1} can be expressed as
\begin{eqnarray}
    \mathcal{I}_p\equiv\mathcal{I}_p(B^1,B^2) = 2^{2 + \frac1p} x\sqrt{1-x^2},
    \label{eq:ipf(x)}
\end{eqnarray}
where we have used the shorthand $\mathcal{I}_p$ for $\mathcal{I}_p(B_1,B_2)$.
Comparing between Eqs.~\eqref{eq:rq} and \eqref{eq:ipf(x)}, we obtain a relation between $\mathcal{R}^Q_2$ and $\mathcal{I}_p$ as 
\begin{eqnarray}
    \mathcal{R}^Q_2 = \frac12 + \frac14 \sqrt{1+ \frac{\mathcal{I}_p}{2^{1+\frac1p}}}.
\end{eqnarray}

\end{document}